\documentclass[conference]{IEEEtran}
\IEEEoverridecommandlockouts
\usepackage{cite}
\usepackage{amsmath,amssymb,amsfonts}
\usepackage{algorithmic}
\usepackage{graphicx}
\usepackage{textcomp}
\usepackage{xcolor}
\def\BibTeX{{\rm B\kern-.05em{\sc i\kern-.025em b}\kern-.08em
    T\kern-.1667em\lower.7ex\hbox{E}\kern-.125emX}}

\usepackage[hyphens]{url}
\usepackage[implicit=false]{hyperref}
\usepackage{booktabs}
\usepackage{amsthm}
\usepackage{multirow}
\usepackage[vlined, ruled, linesnumbered]{algorithm2e}
\usepackage{soul}
\usepackage{tikz}
\usepackage{subcaption}

\newtheorem{theorem}{Theorem}

\newcommand{\code}[1]{\textsf{\small\mdseries #1}}

\newcommand{\anon}[2]{#1}
\newcommand{\ds}{\anon{{\sc DataSpread}\xspace}{{\sc Data\-Spread}\xspace}}

\newcounter{ExpCount}

\SetKwInput{KwData}{Input}
\SetKwInput{KwResult}{Output}

\tikzstyle{cellnode}=[rectangle,inner sep=0pt,draw=none,align=left]
\tikzstyle{c_size_a}=[text width=4.5em,minimum width=5.0em,text height=0.9em,text depth=0.25em]
\tikzstyle{c_size_b}=[inner xsep=0.25em,text height=0.8em,text depth=0.30em,anchor=west,fill=blue!20]
\tikzstyle{c_size_c}=[inner xsep=0.25em,text height=0.8em,text depth=0.30em,anchor=west,fill=red!20]
\tikzstyle{labelnode}=[rectangle,inner sep=0pt,minimum size=0pt,fill=none,draw=none]
\tikzstyle{depgraphnode}=[rounded rectangle,draw,inner sep=2.5pt,minimum size=0pt,minimum width=15pt,minimum height=8pt,fill=gray!20]
\tikzstyle{depgraphrangenode}=[depgraphnode,fill=red!20]
\tikzstyle{depgraphedge}=[->,thick]
\tikzstyle{depgraphbluntedge}=[-,thick]
\tikzstyle{depgraphinhedge}=[-Latex,red,thick,dashed]

\tikzstyle{op_label}=[below=-0.15em]
\tikzstyle{op_a_bar}=[Latex-Latex,very thick]

\newcommand{\squishenum}{
	\begin{enumerate}
		{ \setlength{\itemsep}{-50pt}
			\setlength{\parsep}{0pt}
			\setlength{\topsep}{0pt}
			\setlength{\partopsep}{0pt}
		}
	}
	\newcommand{\squishenumend}{\end{enumerate}}
\newcommand{\stitle}[1]{\vspace{0.2em}\noindent\textbf{#1}}

\newcommand{\squishlist}{
   \begin{list}{$\bullet$}
    { \setlength{\itemsep}{0pt}
      \setlength{\parsep}{2pt}
      \setlength{\topsep}{2pt}
      \setlength{\partopsep}{0pt}
      \setlength{\leftmargin}{12pt}
    }
}
\newcommand{\squishend}{\end{list}}

\newcommand{\cut}[1]{}
\newcommand{\tr}[1]{}

\definecolor{mypurple}{RGB}{119, 69, 198}
\definecolor{newc}{RGB}{70, 180, 80}

\newcommand{\rr}{\textsf{RR}\xspace}
\newcommand{\rf}{\textsf{RF}\xspace}
\newcommand{\fr}{\textsf{FR}\xspace}
\newcommand{\ff}{\textsf{FF}\xspace}
\newcommand{\rrChain}{\textsf{RR-Chain}\xspace}
\newcommand{\rrGap}{\textsf{RR-Gap}\xspace}
\newcommand{\rrGapOne}{\textsf{RR-GapOne}\xspace}
\newcommand{\ncPattern}{\textsf{Single}\xspace}
\newcommand{\naParam}{\textsf{NA}\xspace}
\newcommand{\sys}{TACO\xspace}
\newcommand{\inrow}{TACO-InRow\xspace}
\newcommand{\full}{TACO-Full\xspace}
\newcommand{\anti}{Antifreeze\xspace}
\newcommand{\calc}{NoComp-Calc\xspace}
\newcommand{\nocomp}{NoComp\xspace}
\newcommand{\redis}{RedisGraph\xspace}
\newcommand{\frel}{\textit{rel}\xspace}
\newcommand{\procname}[1]{\texttt{#1}\xspace}
\newcommand{\fcompressname}{addDep}
\newcommand{\fcompress}{\procname{\fcompressname}\xspace}
\newcommand{\ffinddepupdate}{\procname{findDep}\xspace}
\newcommand{\ffindprecupdate}{\procname{findPrec}\xspace}
\newcommand{\fclear}{\procname{removeDep}\xspace}
\newcommand{\varname}[1]{\textit{#1}\xspace}
\newcommand{\pprec}{\varname{prec}\xspace}
\newcommand{\pdep}{\varname{dep}\xspace}
\newcommand{\pmeta}{\varname{meta}\xspace}
\newcommand{\phead}{\varname{head}\xspace}
\newcommand{\ptail}{\varname{tail}\xspace}
\newcommand{\hrel}{\varname{hRel}\xspace}
\newcommand{\hfix}{\varname{hFix}\xspace}
\newcommand{\trel}{\varname{tRel}\xspace}
\newcommand{\tfix}{\varname{tFix}\xspace}

\newcommand{\maxDep}{Maximum Dependents\xspace}
\newcommand{\longPath}{Longest Path\xspace}

\newcommand{\precU}{precToVisit\xspace}
\newcommand{\sprec}{precedent\xspace}
\newcommand{\sprecs}{precedents\xspace}
\newcommand{\sdep}{dependent\xspace}
\newcommand{\sdeps}{dependents\xspace}
\newcommand{\lprec}{precedent\xspace}
\newcommand{\lprecs}{precedents\xspace}
\newcommand{\ldep}{dependent\xspace}
\newcommand{\ldeps}{dependents\xspace}
\newcommand{\pFigure}{Fig.}
\newcommand{\pSection}{Sec.}

\definecolor{darkgreen}{rgb}{0.0, 0.2, 0.13}
\definecolor{auburn}{rgb}{0.43, 0.21, 0.1}
\definecolor{antiquefuchsia}{rgb}{0.57, 0.36, 0.51}
\definecolor{armygreen}{rgb}{0.29, 0.33, 0.13}
\definecolor{ao}{rgb}{0.0, 0.5, 0.0}
\definecolor{purple}{rgb}{0.75, 0.0, 1.0}
\definecolor{yellow}{rgb}{0.99, 0.76, 0.0}

\newif\iftr
\trtrue 

\iftr
\newcommand{\ptr}[1]{#1}
\newcommand{\ppaper}[1]{}
\else
\newcommand{\ptr}[1]{}
\newcommand{\ppaper}[1]{#1}
\fi

\iftr
\newenvironment{reviewone}{\par}{\par}
\newenvironment{reviewtwo}{\par}{\par}
\newenvironment{reviewthree}{\par}{\par}
\newenvironment{reviewmulti}{\par}{\par}
\newcommand{\rone}[1]{#1}
\newcommand{\rtwo}[1]{#1}
\newcommand{\rthree}[1]{#1}
\else
\newenvironment{reviewone}{\par\color{blue}}{\par}

\newcommand{\rone}[1]{\textcolor{blue}{#1}}
\newcommand{\rtwo}[1]{\textcolor{blue}{#1}}
\newcommand{\rthree}[1]{\textcolor{blue}{#1}}
\fi

\usepackage[switch]{lineno}

\newcommand{\comment}[1]{\ignorespaces}

\newenvironment{tightitemize}{\begin{list}{$\bullet$}{\setlength{\rightmargin}{0em}\setlength{\leftmargin}{1em}\setlength{\topsep}{0in}\setlength{\itemsep}{0cm}\setlength{\itemindent}{1em}}}{\end{list}}

\newcommand{\TacoAuthors}{\small Dixin Tang, Fanchao Chen$^1$\textsuperscript{\textsection}, Christopher De Leon, Tana Wattanawaroon$^2$, Jeaseok Yun, Srinivasan Seshadri, Aditya G. Parameswaran}

\begin{document}


\title{Efficient and Compact Spreadsheet Formula Graphs
\vspace{-5mm}}

\author{
\IEEEauthorblockN{\TacoAuthors}
\IEEEauthorblockA{\textit{UC Berkeley $|$ Fudan University$^1$ $|$ UIUC$^2$}\\
\{totemtang, chrisdeleon333, jonathanyun, srinivasan.seshadri, adityagp\}@berkeley.edu, \\ chenfc18@fudan.edu.cn, wattana2@illinois.edu}
\vspace{-12mm}
}

\maketitle
\thispagestyle{plain}
\pagestyle{plain}

\begingroup\renewcommand\thefootnote{\textsection}
\footnotetext{Work done at UC Berkeley.}
\endgroup


\begin{abstract}
Spreadsheets are one of the most popular data analysis tools, 
wherein users can express computation 
as formulae alongside data.
The ensuing dependencies are tracked as formula graphs.
Efficiently querying and maintaining these formula graphs 
is critical for interactivity across multiple settings. 
Unfortunately, formula graphs are 
often large and complex 
such that querying and maintaining them is time-consuming, 
reducing interactivity.
We propose \sys, a framework for efficiently 
compressing formula graphs, thereby 
reducing the time for querying and maintenance. 
The efficiency of \sys stems from a key spreadsheet property: 
tabular locality, which means that cells close to each other are likely 
to have similar formula structures. 
We leverage four such tabular locality-based patterns, 
and develop algorithms for compressing formula graphs 
using these patterns, 
directly querying the compressed graph without decompression, 
and incrementally maintaining the graph during updates. 
We integrate \sys into an open-source 
spreadsheet system and show that 
\sys can significantly reduce formula graph sizes. 
For querying formula graphs, 
the speedups of \sys over a baseline implemented in our 
framework and a commercial 
spreadsheet system are up to 34,972$\times$ and 632$\times$, 
respectively.
\end{abstract}


\section{Introduction}
\label{sec:introduction}

%

Spreadsheets are widely used for data analysis, 
with a user-base of nearly 1 
billion~\cite{excel-users, users-one-billion}. 
They support a variety of applications, 
from planning and inventory tracking,  
to complex financial, medical, and scientific data analysis. 
Their popularity is attributable to 
an intuitive tabular layout  
and in-situ formula computation~\cite{nardi1990spreadsheet}. 
Users directly analyze their data 
by writing embedded formulae alongside data, 
akin to database views~\cite{Gupta93IVM}.
These formulae take the results of 
other formulae or raw data values as input, 
creating dependencies between the output of formulae and 
their inputs. 
These dependencies are internally 
represented 
as a {\em formula graph}. 
\rone{Querying formula graphs 
is critical to the interactivity of spreadsheets for 
multiple applications, including:}

\begin{reviewone}
\begin{tightitemize}
\item \textbf{Formula recalculation:} 
When a cell is modified, 
the spreadsheet system needs to query the formula graph 
to find its dependents and calculate new results
~\cite{excel-dep-graph, calc-dep-graph}. 
In addition, the performance of identifying dependents 
is critical for returning control to users 
interactively for an asynchronous execution model~\cite{Antifreeze}. 
Here, the spreadsheet system marks 
the formulae to be recalculated as invisible 
and returns control to the user such that 
the user can interact with the visible cells. 
Therefore, finding dependents is on the critical path 
for returning control to the user and its performance is 
important in ensuring interactivity. 

\item \textbf{Formula dependency visualization:} 
Spreadsheet systems, such as Excel and LibreCalc, 
provide tools for finding and visualizing the dependents 
and precedents of cells to help users check the accuracy of formulae 
and identify sources of errors~\cite{excel-trace, calc-trace, spreadsheetHorrorStory, spreadsheetError, excel-audit}. 
For these applications, the performance of 
finding dependents/precedents 
is also critical to maintain interactivity.
\end{tightitemize}
\end{reviewone}
\noindent Unfortunately, real-world spreadsheets often include 
large and complex formula graphs, 
where a cell update can have a large number of dependent cells. 
Traversing these graphs to find dependents 
can be time-consuming and lead to high response times. 
We analyze two real-world spreadsheet datasets, 
Enron~\cite{Enron} and Github (a dataset we crawled; 
details in \pSection~\ref{sec:benchmark_config}). 
We compute, for each spreadsheet, 
the maximum number of \sdeps for a given cell, 
as well as the longest path in the formula graph. 
\rone{We plot the probability distributions for these two quantities 
for the two datasets in \pFigure~\ref{fig:num_dependents}.}
We see that the number of dependents of a single cell 
can be as high as 300K, 
while a path can be as long as 200K edges. 
\rone{Therefore, finding dependents and precedents 
in real spreadsheets may take a long response time, 
which in turns hinders data exploration. 
In fact, our experiments in Sec.~\ref{sec:exp} show that 
using a baseline approach to find dependents in formula graphs 
can take up to 49 seconds for real spreadsheets. 
A previous study has shown that even an additional delay 
of 0.5 seconds ``results in reduced interaction and reduced 
dataset coverage during analysis''~\cite{Latency_LiuH14}.}


\begin{figure}[!t]
    \centering
    \includegraphics[height=26mm]{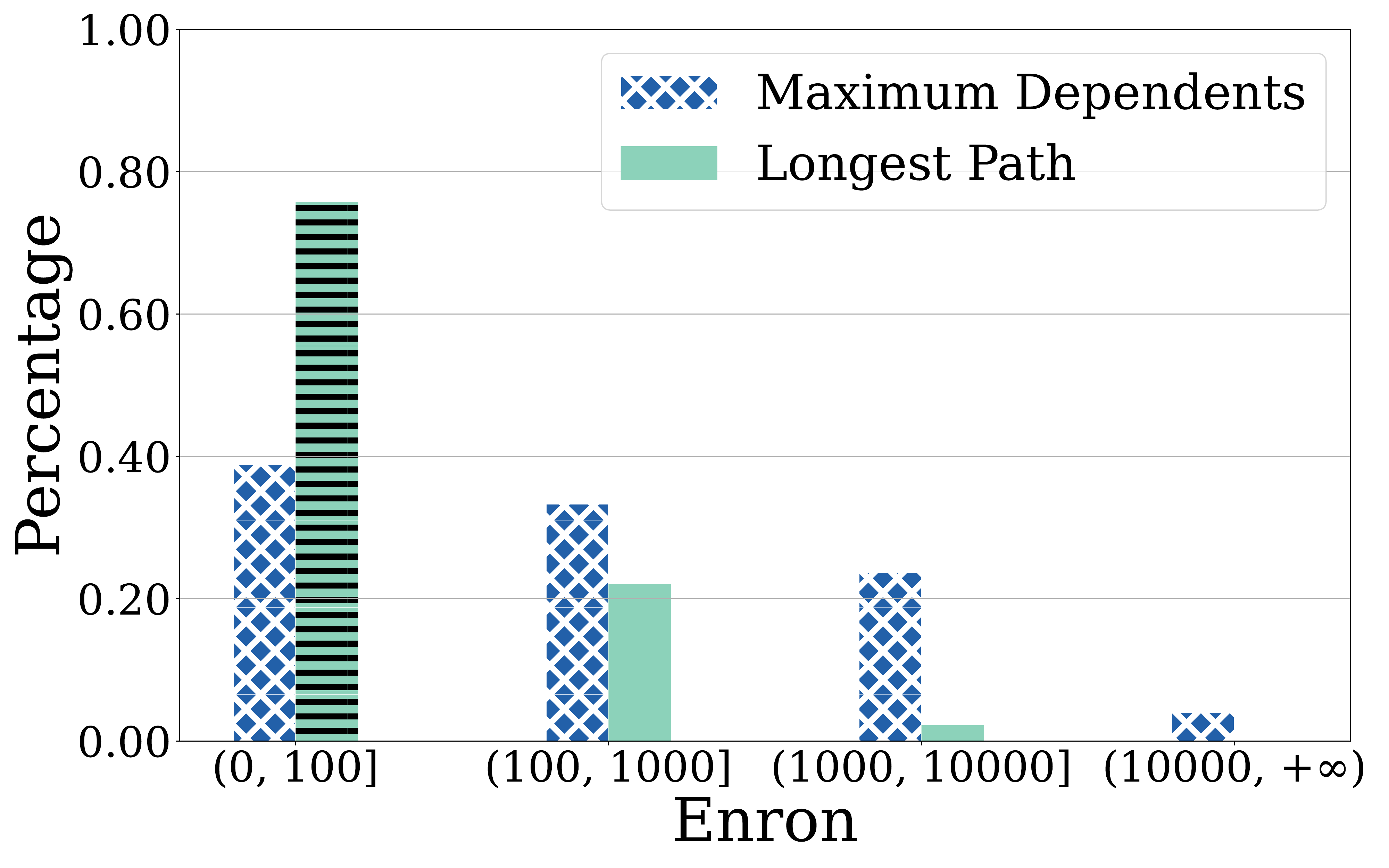}
    \includegraphics[height=26mm]{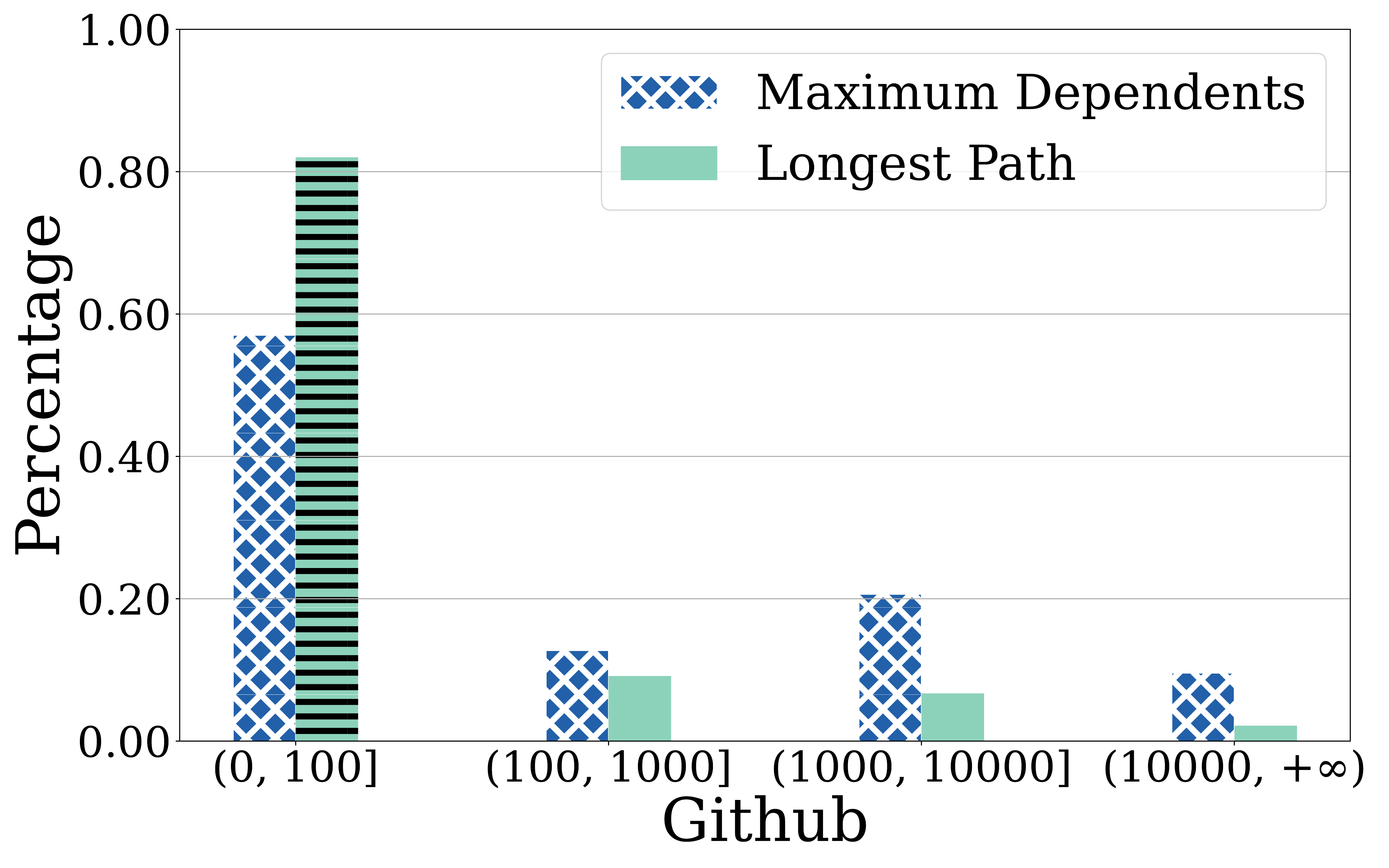}
    \vspace{-3mm}
    \caption{\small \rone{Probability distributions for 
    the maxinum number of \sdeps and the longest path 
    in the Enron and Github datasets}}
    \label{fig:num_dependents}
    \vspace{-3mm}
\end{figure}

\begin{figure}[!t]
    \centering
    \includegraphics[height=24mm]{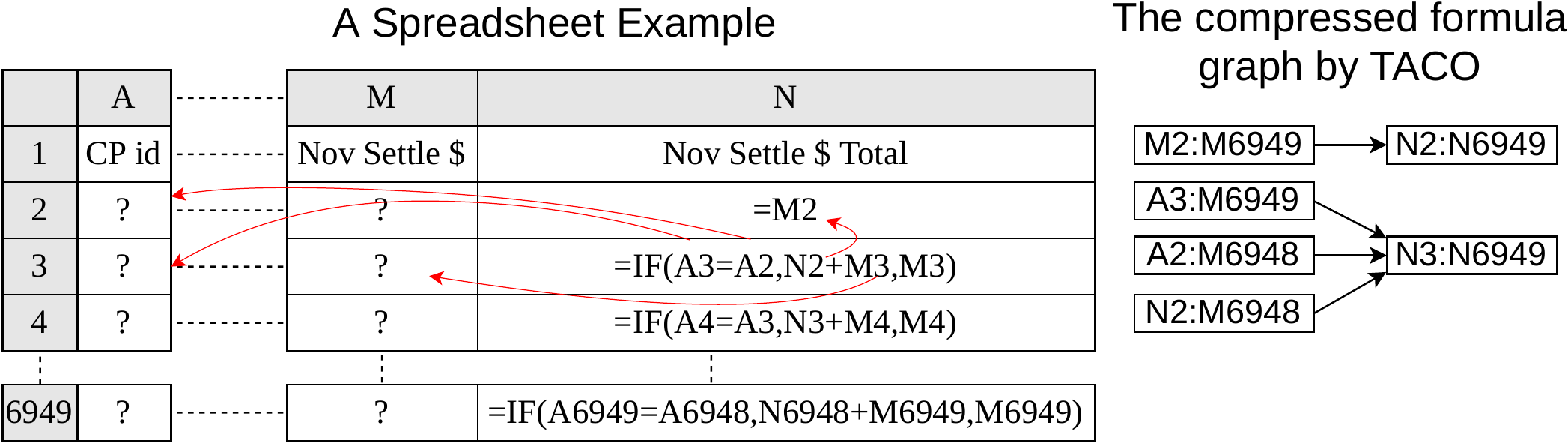}
    \vspace{-2mm}
    \caption{\small A real spreadsheet with tabular locality~\cite{Enron}}
    \label{fig:locality_example}
    \vspace{-7mm}
\end{figure}

\rone{In this paper, we study how we can reduce the execution time 
for finding dependents and precedents in formula graphs 
by leveraging predefined patterns found in real spreadsheets.}
These patterns can be used to 
compress the formula graph 
and enable fast look-ups of dependents and precedents 
directly on the compressed graph. 
\emph{Specifically, our key insight is that cells 
that are close to each other 
in the tabular spreadsheet 
layout often employ similar formula structures}, 
a property we refer to as \textbf{tabular locality}. 
\pFigure~\ref{fig:locality_example} shows a column of formulae 
of a real-world spreadsheet that follows tabular locality. 
While the formulae in the column \code{N} 
look complicated at first glance, 
they follow the same pattern starting from \code{N3}: 
the \code{IF} formula in each row 
references the cell of the same row and the row above 
from column \code{A} 
(e.g., \code{N3} references \code{A3} and \code{A2}), 
the cell to the left (e.g., \code{M3} for \code{N3}), 
and the cell above (e.g., \code{N2} for \code{N3}). 
Tabular locality is prevalent in real-world spreadsheets mainly 
because users often do not 
write several distinct formulae by hand, 
but use spreadsheet features, 
such as copy-paste and autofill, 
to generate formulae automatically. 
Autofill, for instance, allows users 
to drag a cell to fill adjacent cells 
by repeating the pattern of the source cell. 
Users could also programmatically 
generate a large spreadsheet, 
which still likely respects tabular locality. 


\rone{However, leveraging tabular locality 
to efficiently find dependents and precedents requires 
addressing several challenges in 
compressing, querying, and maintaining formula graphs.}
First, as we show in the experiments (\pSection~\ref{sec:exp_graph_size}), 
formula references are complex: 
a formula could include multiple references 
to different rows or columns 
(e.g., \code{N3} in \pFigure~\ref{fig:locality_example} 
references 4 different cells). 
Disentangling these multiple references across cells, 
identifying common patterns across them 
and compressing these common patterns 
can be time-consuming. 
Therefore, the compression algorithm should balance 
between the quality of the compression 
(e.g., the number of common patterns detected) 
and the compression time. 
While it may be possible to track user actions 
(e.g., during autofill) to identify and compress 
formula dependencies directly, 
this approach does not 
apply when spreadsheets are shared via files 
(e.g., xlsx files), losing track of the actions that generated them. 
It also cannot compress formula dependencies generated programmatically~\cite{excel-vba-dependents, poi}.
Furthermore, developing a compressed representation and 
algorithms for directly querying the compressed graphs 
to reduce look-up time is also an open challenge
Finally, the formula graph needs to be maintained over time, 
which requires incrementally updating the compressed graph 
to avoid decompression overhead.
 


To address these challenges, we present 
\emph{\underline{Ta}bular Locality-based 
\underline{Co}mpression} or \sys. 
In \sys, we leverage four basic patterns that serve as building blocks 
for other more complicated patterns, and identify one extended pattern 
based on the analysis of real spreadsheets. 
We propose a generic framework that 
decomposes messy formulae and extracts their predefined patterns. 
This framework is also extensible to support new patterns.
We prove that compressing a formula graph using predefined patterns 
is {\sc NP-Hard} via a reduction from the 
\emph{rectilinear picture compression} problem~\cite{intractability}. 
We develop a greedy algorithm to 
efficiently compress formula graphs 
while maintaining low compression overhead.
Further, we design algorithms for 
finding \ldeps or \lprecs directly on the compressed graph, 
and for maintaining the graph incrementally, 
and analyze the complexity of each algorithm.
Our experiments show that for querying formula graphs, 
the speedups of \sys over a baseline  
and a commercial spreadsheet system are 
up to \textbf{34,972$\times$} 
and \textbf{632$\times$}, respectively.

While there is a lot of work on graph compression~\cite{liu2018graph}, 
this work does not leverage tabular locality or take into account 
the spatial nature of spreadsheet ranges. 
In addition, most of this work does not support directly 
querying the compressed graph, so these compression 
algorithms will not be faster than an approach 
without compression in terms of finding \sdeps/\sprecs.
\ptr{Fan et al.~\cite{queryPreserving} propose a method for directly 
executing reachability and pattern matching queries 
on a compressed graph, but do not leverage tabular locality 
nor supporting finding \sdeps/\sprecs.}
\sys is also different from columnar compression~\cite{MonetDB, CStore} because \sys compresses formula dependencies 
as opposed to columnar data. 
A recent paper proposes a specialized algorithm 
for compressing formula graphs~\cite{Antifreeze}.
However, this algorithm introduces false positives.
In addition, it has a high compression and maintenance time, 
which we will show in \pSection~\ref{sec:exp}. 
It turns out that Excel has a capability 
wherein it identifies identical formulae and 
stores duplicate formulae as  
pointers to the first formula~\cite{excel_same_formula}. 
But it does not consider compressing and querying 
the formula dependencies based on tabular locality.



\section{Background and Problem}
\label{sec:background}

In this section, we present the problem of compressing, querying, 
and maintaining formula graphs. 

\subsection{Background}
\label{sec:formula_graph}

\stitle{Spreadsheets} 
A spreadsheet consists of a set of \emph{cells} organized in a tabular layout. 
Each cell is referenced using its column and row index. 
Columns are identified by letters \code{A,} $\cdots$ \code{, Z, AA,} $\cdots$, 
and rows are identified by numbers \code{1, 2,} $\cdots$. 
For simplicity, for a cell we also use integers $(i, j)$ 
to represent its position, 
where $i$ and $j$ are column and row indices,
respectively, both starting from 1. 
A \emph{range}, akin to a 2D window, 
is a rectangular region of cells, 
identified by the top-left (called \emph{head}) and 
bottom-right (called \emph{tail}) cells. 
For example, the range \code{A1:B2} contains cells \code{A1, A2, B1, B2},
with 
head and tail cells \code{A1} and \code{B2}, respectively. 

A cell contains a \emph{formula} or a \emph{pure value}. 
A pure value is a constant belonging to a fixed type 
while a formula is a mathematical expression that 
takes pure values and/or cell/range references as input. 
The result of an evaluated formula is an \emph{evaluated value}. 
For example, the cells in column \code{A} in \pFigure~\ref{fig:locality_example} 
include pure values while the cells in column \code{N} include formulae. 
In the rest of the paper, we use ``value'' to refer to either 
the pure or evaluated value of a cell. 
A cell that includes a formula is called a \emph{formula cell}.

\stitle{Formula graphs} 
A formula graph is a directed acyclic graph (DAG) 
that stores the dependencies of each formula 
referencing other ranges as edges. 
Specifically, each formula is parsed to 
get the set of ranges the formula references, 
with a directed edge added from each referenced range to 
the formula cell. 
We call this directed edge a \emph{dependency}. 
Given a directed edge $e = (\pprec, \pdep)$, 
we call $\pprec$ the \emph{direct precedent} of $\pdep$ 
or the \emph{precedent} of the edge $e$. 
Symmetrically, we call $\pdep$ the \emph{direct dependent} 
of $\pprec$ or the \emph{dependent} of $e$. 

\pFigure~\ref{fig:graph_example} shows a spreadsheet 
with four formulae along with its formula graph. 
The cells denoted as $?$ are pure values. 
\code{B1} and \code{B2} have the same formula \code{SUM(A1:A3)}, 
so the direct \sdeps of \code{A1:A3} include \code{B1} and \code{B2}.
\code{C1} references \code{B1} and \code{B3}, so we add two edges.
Finally, \code{C2} references \code{B2:B3}, which adds an edge 
with separate vertices although \code{B2:B3} 
overlaps with the vertices \code{B2} and \code{B3}.

The formula graph is used to quickly find the \emph{\sdeps}
or \emph{\sprecs} of an input range. 
The \sdeps of an input range 
are the set of cells that are reachable from 
the input range in the formula graph. 
Symmetrically, the \sprecs are 
the set of cells that can reach the 
input range in this formula graph via a directed path. 
For example, the \sdeps of \code{A1} are \code{\{B1, B2, C1, C2\}} 
in \pFigure~\ref{fig:graph_example}. 
Since a vertex in the graph can be a range, 
we can build an index (e.g., R-Tree~\cite{r-tree})
over the vertices to quickly find the ranges that 
overlap with an input range.
(e.g., find \code{A1:A3} given a cell \code{A1}).

One application of the formula graph 
is to find the \sdeps when users update 
the spreadsheet to ensure that users do not see stale 
or inconsistent results. Specifically, when a cell is updated, 
the formulae of its \sdeps 
will be re-evaluated in sequence to refresh their values. 
A key prerequisite to updating the formulae is 
identifying which formulae require recomputation in the first place. 
If the update is to a formula cell, 
the formula graph will be modified. 
The formula graph is also useful for 
visualizing formula dependencies, 
which allows users to trace the \sdeps/\sprecs of a cell 
to check the accuracy of formulae or 
find the sources of errors~\cite{excel-audit, calc-trace, excel-trace, excel-vba-dependents}. 
In both applications, the performance of querying the formula graph 
is critical to ensure interactivity.


\subsection{Compressing, Querying, and Maintaining Formula Graphs}

\stitle{Formula graph problems} 
Finding the \sdeps and \sprecs of an input range 
is often time-consuming due to the size 
or complexity of the graph. 
Therefore, we propose compressing formula graphs  
by leveraging tabular locality 
to significantly reduce graph size.
Directly querying the compressed graph 
and incremental maintenance 
can decrease the time taken for finding \sdeps/\sprecs 
and maintaining the graph, respectively,  
while introducing the modest overhead of 
building the compressed graph.

\stitle{Patterns in formula graphs} 
We capture and distill tabular locality 
in a formula graph via {\em patterns}. 
For a set of edges $A$ of arbitrary size, 
a pattern is a constant-size (i.e., $O(1)$) representation 
that can reconstruct $A$ (with size $O(|A|)$). 
In addition, finding the direct \sdeps 
or \sprecs of an input range in a pattern should also 
be constant time.
Consider \pFigure~\ref{fig:locality_example} as an example. 
Each formula cell \code{N$i$} starting from \code{N3} follows 
the pattern that \code{N$i$} depends on \code{A$i$}, 
\code{A($i$-1)}, \code{M$i$}, and \code{N($i$-1)}. 
By storing the relative positions between \code{N$i$} and \code{A$i$}, i.e., \code{A$i$} is 13 columns left to \code{N$i$}, 
and the valid range of \code{N$i$}, i.e., \code{N3:N6949}, 
we can represent the edges $\code{A$i$} \rightarrow \code{N$i$}$ 
using constant-size information. 
We can also find \sdeps/\sprecs in this compressed 
edge in constant time. For example, for input range \code{A3:A10}, 
we can use the information of relative positions 
to find the \sdeps \code{N3:N10} in constant time. 
We can represent and query other edges in a similar way. 
Therefore, leveraging patterns greatly reduces formula graph sizes 
and consequently reduces the time for finding \sdeps/\sprecs. 

\begin{figure}[!t]
    \centering
    \includegraphics[height=18mm]{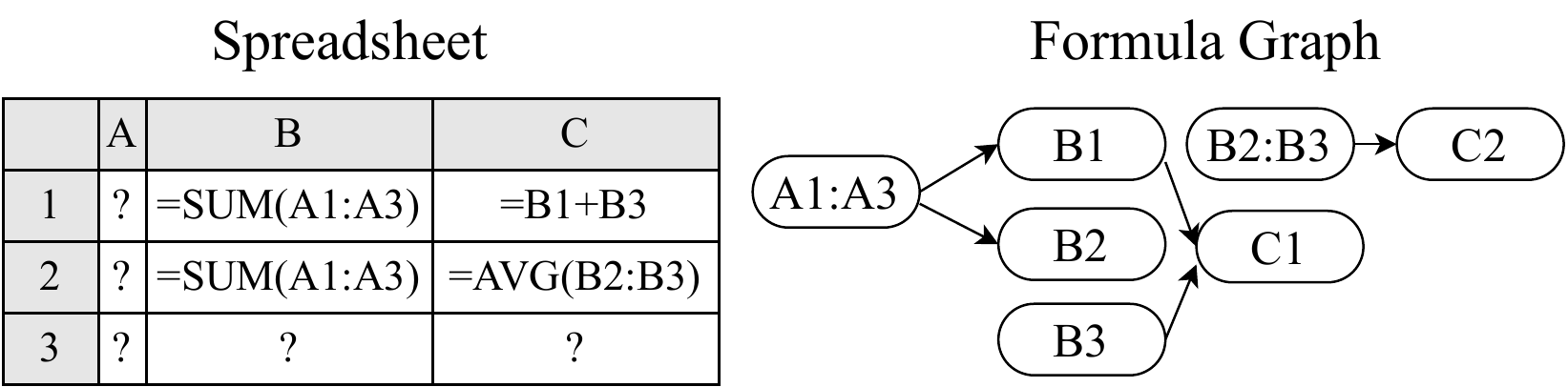}
    \caption{An example spreadsheet and its formula graph}
    \label{fig:graph_example}
    \vspace{-7mm}
\end{figure}

\stitle{Compressed formula graph representation} 
Given a formula graph $G^\prime(E^\prime, V^\prime)$, 
we want to find a \emph{lossless} compressed graph $G(E, V)$ 
that preserves the results of 
finding \sdeps/\sprecs. 
$G$ is generated 
based on a partition of the edge set $E^\prime$ in $G^\prime$: 
$P=\{E_1^\prime, E_2^\prime, \cdots, E_N^\prime\}$, 
where $E^\prime=\cup_{i=1}^{N}E_i^\prime$. 
Each $E_i^\prime$ must either contain a single edge 
(i.e., uncompressed), or be a set of edges that
follow one of the predefined compression patterns, in which case
these edges will be replaced with a compressed edge. 
We generate a compressed edge $e_i=(\pprec, \pdep, p, \pmeta)$ 
for each $E_i^\prime = \{e_1^i, e_2^i, \cdots, e_{M_i}^i\}$, 
where $M_i$ is the size of $E_i^\prime$. 
The components $e_i.\pprec$ and $e_i.\pdep$ are the precedent 
and the dependent of $e_i$, computed 
as $e_i.\pprec = \bigoplus_{j=1}^{M_i}{e_j^i.\pprec}$ and 
$e_i.\pdep = \bigoplus_{j=1}^{M_i}{e_j^i.\pdep}$. 
$\bigoplus$ is the minimal bounding range 
of the input ranges. 
For example, $\bigoplus$ merges the ranges \code{A1:A3} 
and \code{A2:A5} into \code{A1:A5}.
The component $e_i.p$ is the compression pattern 
while $e_i.\pmeta$ encodes the underlying pattern 
information such that $\pmeta$ of an edge $e_i$ in $E$ 
is used to reconstruct the corresponding edges in $E_i^\prime$. 
\rtwo{If $E_i^\prime$ contains a single edge,  $e_i.p$ is set to 
\ncPattern, which is defined as the pattern for an uncompressed edge.}
So $E$ in $G$ comprises the edges generated 
from the partition $P$ and $|E| = N$. 
The compressed vertex set $V$ is induced from $E$.

\begin{figure*}[!t]
\vspace{-20pt}
  \begin{tabular}{cc}
      \begin{minipage}[b]{0.47\linewidth}
        \centering
        \includegraphics[width=0.9\linewidth]{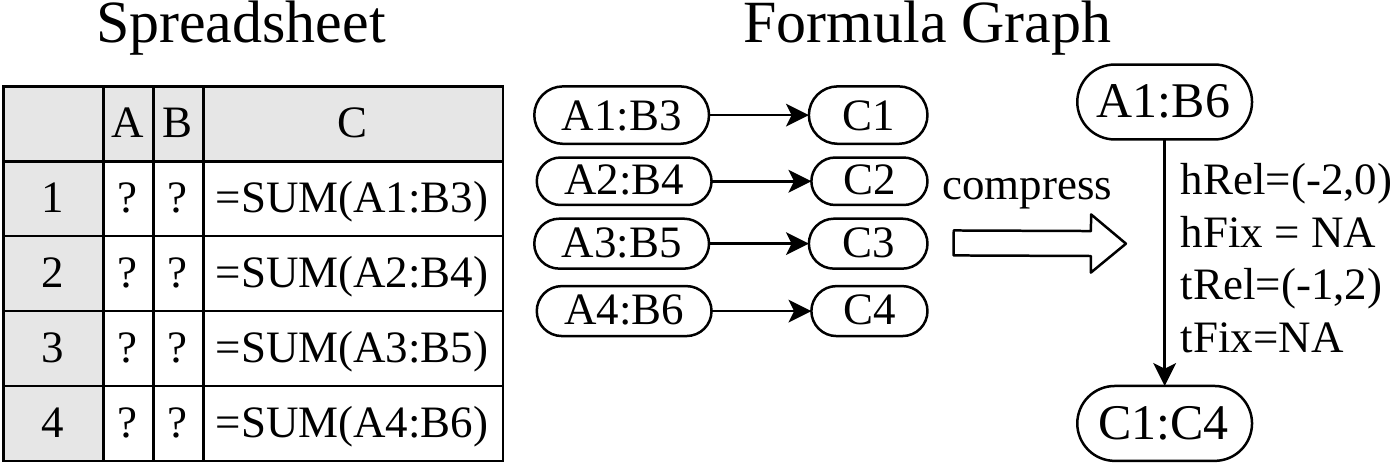}
        \vspace{-2mm}
        \subcaption{Relative plus Relative (i.e., sliding window)}
        \label{fig:rr}
      \end{minipage} 

      \begin{minipage}[b]{0.47\linewidth}
        \centering
        \includegraphics[width=0.9\linewidth]{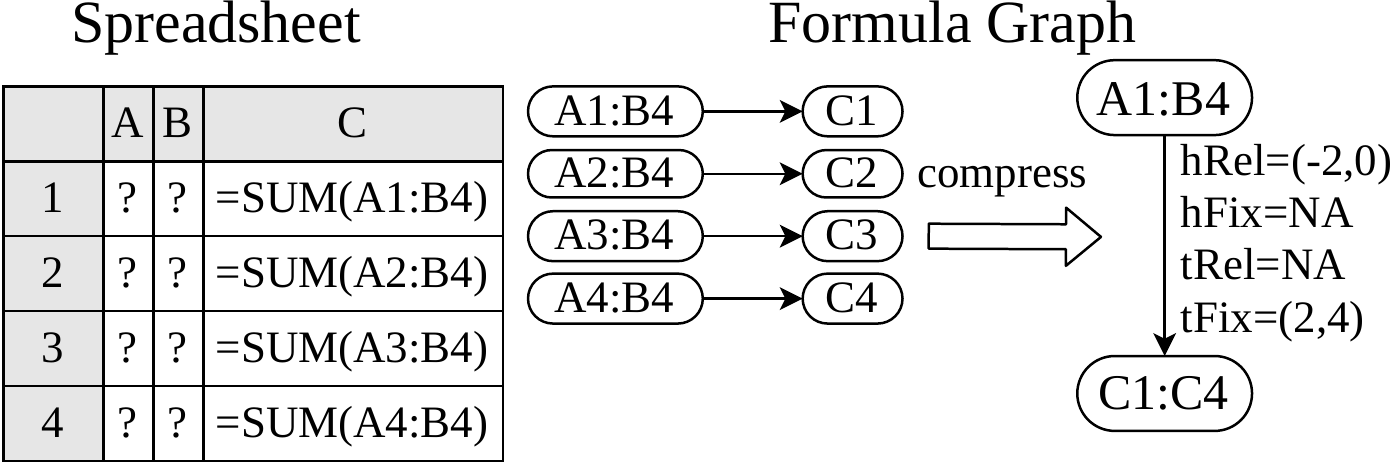}
        \subcaption{Relative plus Fixed (i.e., shrinking window)}
        \label{fig:rf}
      \end{minipage} \\

      \begin{minipage}[b]{0.47\linewidth}
        \centering
        \includegraphics[width=0.85\linewidth]{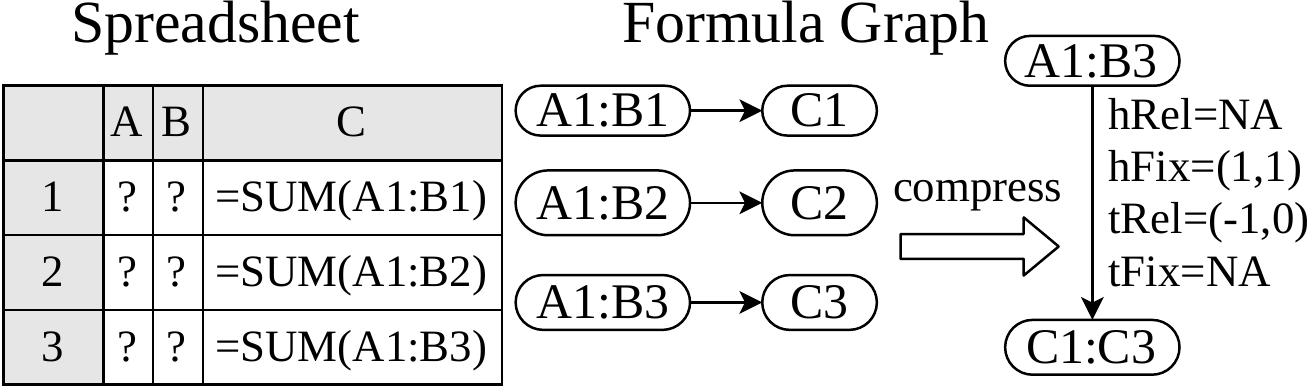}
        \subcaption{Fixed plus Relative (i.e., expanding window)}
        \label{fig:fr}
      \end{minipage} 
      
      \begin{minipage}[b]{0.47\linewidth}
        \centering
        \includegraphics[width=0.85\linewidth]{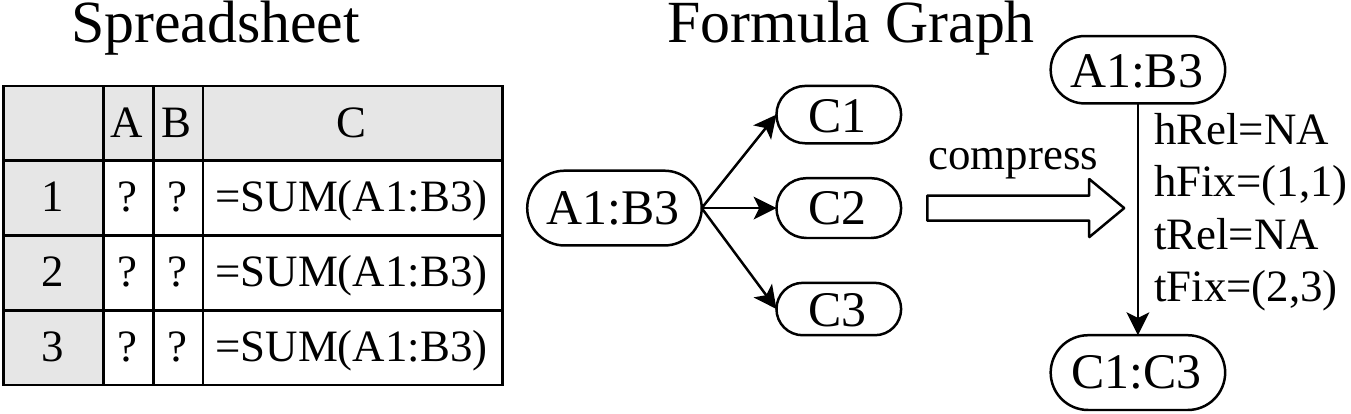}
        \subcaption{Fixed plus Fixed (i.e., fixed window)}
        \label{fig:ff}
      \end{minipage}
  \end{tabular}
  \vspace{-3mm}
  \caption{Examples of four basic patterns of tabular locality}
  \vspace{-7mm}
\end{figure*}

\stitle{Problem statement} 
Given a set of predefined pattern types, we want to 
build an equivalent compressed graph $G(E, V)$ of 
a formula graph $G^\prime(E^\prime, V^\prime)$ 
such that the size of $G$ and the 
time for finding the \sdeps/\sprecs
of a cell are both significantly smaller compared to $G^\prime$, 
and \rone{the time for maintaining $G$ is also 
smaller than the time for maintaining $G^\prime$  
with respect to the same updates}. 

To address this problem, we discuss the 
basic compression patterns we 
support (\pSection~\ref{sec:locality}). 
We then present the \sys framework 
that leverages the basic patterns to compress, query, 
and maintain a formula graph, 
and analyze the algorithmic complexity (\pSection~\ref{sec:sys}). 
Finally, we discuss extending \sys to support new patterns 
beyond the basic ones (\pSection~\ref{sec:extension}).
\section{Basic Patterns}
\label{sec:locality} 

In this section, we define the basic compression 
patterns we consider and propose algorithms for using 
the pattern to build a compressed edge, 
finding \ldeps or \lprecs, and maintaining the edge. 

\subsection{Basic Patterns}
\label{sec:basic_patterns}
The basic patterns consider 
tabular locality in adjacent formula cells 
and assume each formula cell references a single range. 
We can employ the basic patterns multiple 
times to compress dependencies 
when formula cells have multiple references, 
as discussed in \pSection~\ref{sec:compress}. 
We focus on adjacent cells in a column for simplicity; 
the row-wise case can be derived symmetrically. 

Let $A^\prime (A^\prime \subseteq E^\prime $) be a set 
of dependencies in a column of adjacent formula cells. 
Our basic patterns capture various relationships 
between the \lprec and \ldep of 
each $e^\prime \in A^\prime$. 
Recall that $e^\prime.\pdep$ is a formula cell 
and $e^\prime.\pprec$ is a referenced range, represented 
by the positions of its head and tail cell. 
In a spreadsheet, there are two types of relationships 
between the formula cells (i.e., $e^\prime.\pdep$)
and the head/tail cells in the referenced range (i.e., $e^\prime.\pprec$): 
\emph{fixed} and \emph{relative}~\cite{fixed_and_relative}. 

The fixed relationship captures the scenario
where each dependent $e^\prime.\pdep$ 
references the same head or tail cell of 
the precedent $e^\prime.\pprec$. 
For example, the formulae in a column all 
reference a common dollar conversion rate, 
stored in a fixed location, i.e., a cell. 
Here, we have a fixed relationship with both the head and the tail cell of the referenced range, which are identical. 
This is a case of an FF (or Fixed-Fixed) pattern.
The referenced head and tail cell for fixed references
are denoted $\hfix$ and $\tfix$, respectively. 

The relative relationship, on the other hand, 
captures the scenario where each dependent $e^\prime.\pdep$ 
has the same relative position 
with respect to the head or tail cell 
of the precedent $e^\prime.\pprec$. 
The relative position with respect 
to the head and tail cell 
is denoted $\hrel$ and $\trel$, respectively. 
One example of a relative position is 
each formula cell in a column 
references a cell to its left. 
Here there is a relative relationship
to both the head and tail cell of the referenced range,
which are identical. 
This is a case of an RR (or Relative-Relative) pattern.
We use a pair $(p, q)$ to represent the relative position; 
$p$ denotes the relative column distance 
and $q$ the relative row distance. 
Given two cells' positions $u$ and $v$, 
we say $u$ is {\em relative} to $v$ by $(p, q)$ 
if $v.i = u.i + p$ and $v.j = u.j + q$.

Combining the two types of relationships (fixed or relative) with the two cells that represent the \lprec (head and tail), 
there are four basic patterns that capture
the relationships between \ldeps and \lprecs
for a column of dependencies.
We additionally include a 
default pattern called \ncPattern 
for an uncompressed edge.

\stitle{Relative plus Relative (\rr)}
\rr is the setting where each $e^\prime.\pdep$ has the same relative positions to 
both head and tail cells of $e^\prime.\pprec$. 
\pFigure~\ref{fig:rr} shows an example. 
We see that each formula cell in column \code{C} is 
relative to the head cell of its referenced range 
by $(-2, 0)$ (i.e., to the left by two columns) 
and relative to the tail cell by $(-1, 2)$. 
The metadata $\pmeta$ is 
($\hrel=(-2,0), \hfix=\naParam, \trel=(-1,2), \tfix=\naParam)$). 
$\naParam$ means that this pattern does not include this information 
(e.g., \rr does not reference fixed head or tail cells). 
So the compressed edge in \pFigure~\ref{fig:rr} is  
$(\pprec=\code{A1:B6}, \pdep=\code{C1:C4}, p=\rr, \pmeta)$, 
where $\pmeta$ is as defined above.

\stitle{Relative plus Fixed (\rf)}
\rf is the setting where each $e^\prime.\pdep$ has the same relative position 
to the head cell and references a fixed tail cell. 
\pFigure~\ref{fig:rf} shows an example. 
Here, each formula in column \code{C} is relative to the head cell 
of its referenced range 
by $(-2,0)$ and points to a fixed tail cell $\code{B4} = (2,4)$. 
The metadata $\pmeta$ equals $(\hrel=(-2,0), \hfix=\naParam, \trel=\naParam, \tfix=(2,4))$. 

\stitle{Fixed plus Relative (\fr)}
\fr is the dual pattern of \rf. So each $e^\prime.\pdep$ points to the same head cell 
and has the same relative position for the tail cell. 
\pFigure~\ref{fig:fr} shows that 
the metadata of the compressed edge is $(\hrel=\naParam, \hfix=(1,1), \trel=(-1,0), \tfix=\naParam)$.

\stitle{Fixed plus Fixed (\ff)}
\ff is the setting where each $e^\prime.\pdep$ references fixed head and tail cells, 
represented as $\hfix$ and $\tfix$, respectively. 
\pFigure~\ref{fig:ff} shows that each formula cell always 
points to \code{(A1:B3)} and its $\pmeta$ is $(\hrel=\naParam, \hfix=(1,1), \trel=\naParam, \tfix=(2,3))$.

\stitle{\rtwo{Applicability of the  basic patterns}}
\rtwo{One major reason why the basic patterns are prevalent 
in real spreadsheets is that modern spreadsheet systems 
(e.g., Excel, Google Sheets, and LibreOffice Calc)
provide a tool, called autofill, to help users generate 
a large number of formulae automatically, 
and the patterns adopted by autofill end up being the basic patterns. 
Specifically, autofill generates formulae by 
applying the pattern of one source formula cell to adjacent cells. 
For the generated formulae, the formula functions (e.g., \code{SUM}) 
are the same as in the source cell, but the references 
(e.g., \code{A1:B1}) are modified based on the following rules: 
if a reference is prefixed 
with a dollar sign \texttt{\$}, it is a fixed reference; 
otherwise it is a relative reference. 
Therefore, if a range $R$ in the source cell 
does not have \texttt{\$}, 
the generated ranges from $R$ follow the \rr pattern. 
If $R$'s head cell does not have \texttt{\$} but the tail cell does, 
the generated ranges follow \rf. 
The \fr pattern is generated symmetrically. 
Finally, if $R$'s tail and head cells have \texttt{\$}, 
the generated ranges follow \ff. 
But a formula may mix many ranges and include outliers, 
which makes extracting basic patterns from these formulae challenging.}

\rtwo{
While there are other methods for generating formula cells 
(e.g., programmatically), we believe these methods are still likely to 
generate the basic patterns 
since these patterns are the building blocks for many applications. 
A concrete example that follows \rr  
is in \pFigure~\ref{fig:locality_example};
\rr is common in sliding-window-style computation. 
\fr and \rf are often employed for cumulative total computation.
A real example involves a user sorting transactions 
by date and using a column of formulae to compute 
the year-to-date sales amount. \ff is also widely used 
in real applications for referencing fixed ranges 
for point lookups, e.g., a fixed interest rate or 
monetary conversion rate in a cell, or range lookups. 
One example of a \ff range lookup is a column of \code{VLOOKUP} 
formulae, each of which looks up a value in the same range. 
In fact, our experiments in Section~\ref{sec:exp_graph_size} 
show that leveraging the basic patterns 
in Section~\ref{sec:extension} reduce the number of 
edges of formula graphs for two real spreadsheet datasets 
to 5\% and 1.9\%, respectively, which shows that the 
basic patterns are prevalent in real spreadsheets.}

\ptr{
\stitle{Discussion on leveraging the basic patterns} 
One natural question is whether existing spreadsheets leverage 
the basic patterns for compression and querying. 
Excel does identify the same formulae and stores them 
efficiently~\cite{excel_same_formula}, but does not 
leverage these patterns to accelerate traversing formula graphs, 
as is verified by our experiments in \pSection~\ref{sec:exp_commercial}. 
While it is possible to track autofill expressions to compress 
formula dependencies, this approach does not apply to spreadsheets 
generated programatically, and is coupled 
with a spreadsheet system.
To the best of our knowledge, 
no spreadsheet systems compress and query formula dependencies 
via tabular locality.}


Therefore, we develop general compression algorithms 
that apply to all spreadsheets. 
Our algorithms can certainly leverage user actions (e.g., autofill) 
or cues (e.g., dollar sign)
for better compression, but do not rely on it. 
In addition, we also design novel algorithms 
for efficiently querying and incrementally 
maintaining these patterns.

\subsection{Algorithms for the Basic Patterns}
\label{sec:keyfunctions}
To integrate a pattern 
into \sys, \sys requires each pattern to
implement four key functions\ptr{, as shown in \pFigure~\ref{fig:api}}. 

\begin{tightitemize}
\item $\fcompress(e, e^\prime)$: add a dependency $e^\prime$ 
to a compressed edge~$e$, where $e^\prime.\pdep$, the formula cell, is adjacent to $e.\pdep$;

\item $\ffinddepupdate(e, r)$: find the \ldeps of a range of cells $r$ 
within a compressed edge $e$, where $r$ is contained in $e.\pprec$; 

\item $\ffindprecupdate(e, s)$: find the \lprecs of a 
range of cells $s$ 
within a compressed edge $e$, where $s$ is contained in $e.\pdep$;

\item $\fclear(e, s)$: remove the dependencies 
for a range of formula cells $s$ in an edge $e$, 
where $s$ is contained in $e.\pdep$ \end{tightitemize}
The parameter assumptions are guaranteed by \sys framework, discussed in \pSection~\ref{sec:sys}. 

\ptr{
\begin{algorithm}[!t]
  \small
  \SetAlgoLined\DontPrintSemicolon
  \SetArgSty{textrm}
  \SetKw{break}{break}
  \SetKwProg{myalg}{Algorithm}{}{}
  \SetKwProg{myproc}{Procedure}{}{}
  
  \SetKwFunction{compress}{\fcompressname}
  \SetKwFunction{rel}{rel}
  \myalg{\compress{$e$, $e^\prime$}}{
    \uIf{$e.p == \ncPattern$}{
        \uIf{$\frel(e) == \frel(e^\prime)$}{
            \Return $(e.\pprec \bigoplus e^\prime.\pprec, e.\pdep \bigoplus e^\prime.\pdep, \rr, \frel(e^\prime))$\;
        }
    } \ElseIf {$e.\pmeta == \frel(e^\prime)$} {
        \Return $(e.\pprec \bigoplus e^\prime.\pprec, e.\pdep \bigoplus e^\prime.\pdep, \rr, \pmeta)$\;
    }
    \Return NULL\;
  }{}
  \myproc{\rel{$e$}}{
     $hRel \gets e.prec.head - e.dep$\;
     $tRel \gets e.prec.tail - e.dep$\;
    \Return  $(hRel, tRel)$
  }{}
  
  \BlankLine
  \SetKwFunction{query}{findDep}
  \SetKwFunction{findHead}{findDepHead}
  \SetKwFunction{findTail}{findDepTail}
  \myalg{\query{$e, r$}}{
    $\pprec_t^{d_h} \gets (e.\pprec.\ptail.i, \varname{r.\phead}.j)$\;
    $d_h \gets \pprec_t^{d_h} - e.\pmeta.\trel$\;
    $\pprec_h^{d_t} \gets (e.\pprec.\phead.i, \varname{r.\ptail}.j)$\;
    $d_t \gets \pprec_h^{d_t} - e.\pmeta.\hrel$\;
    \Return the intersection of $(d_h, d_t)$ and $e.\pdep$\;
  }{}
  
  \SetKwFunction{findPrec}{findPrec}
  \myalg{\findPrec{$e, r$}}{
     $g_h \gets r.\phead + e.\pmeta.\hrel $\;
     $g_t \gets r.\ptail + e.\pmeta.\trel $\;
    \Return $(g_h, g_t)$; 
  }
  
  \BlankLine
  
  \SetKwFunction{clear}{removeDep}
  \myalg{\clear{$e, s$}}{
    $\varname{newDepSet} \gets$ \text{delete} $s$ \text{from} $e.\pdep$\;
    \For{$\varname{newDep} \in \varname{newDepSet}$}{
        $\varname{newPrec} \gets \ffindprecupdate(e, \varname{newDep})$\;
        $p \gets |\varname{newDep}| == 1 \:?\: \ncPattern : \rr$\;
        \text{add} $(\varname{newPrec}, \varname{newDep}, p, e.\pmeta)$ to $\varname{retSet}$\;
    }
    \Return $\varname{retSet}$
  }{}
  
  \caption{Algorithms for the \rr pattern}
  \label{alg:rr}
\end{algorithm}}

\ptr{
\begin{figure}[!t]
\vspace{-10pt}
    \centering
    \includegraphics[height=30mm]{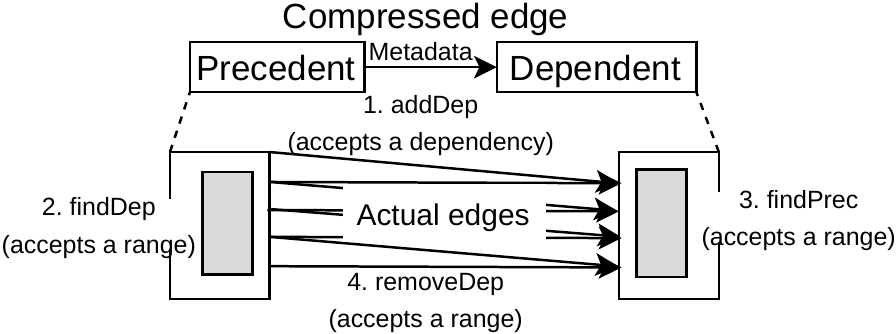}
    \vspace{-1mm}
    \caption{An overview of the four key functions}
    \label{fig:api}
    \vspace{-6mm}
\end{figure}
}

\stitle{\rr} 
\ptr{The four key functions
for \rr are shown in Algorithm~\ref{alg:rr}.}
Consider $\fcompress(e, e^\prime)$, 
which adds a dependency $e^\prime$
to a compressed edge $e$. 
According to the definition of \rr,
the dependency $e^\prime$ can only be added to $e$
if the relative positions of $e^\prime$
equal the relative positions $(\hrel,\trel)$ in $e.meta$.
We define $\frel(e^\prime)$ as the procedure to 
compute the relative positions of $e^\prime.\pdep$ with respect to 
the head and tail cell of $e^\prime.\pprec$, respectively 
\ptr{(i.e., line 9-12 in Algorithm~\ref{alg:rr})}. 
For example, if $e^\prime = \code{A5:B7} \rightarrow \code{C5}$, 
the relative positions are $\hrel=\code{A5}-\code{C5}=(-2,0)$ 
and $\trel=\code{B7}-\code{C5}=(-1,2)$. 
This dependency can be added to the compressed edge $e$ in \pFigure~\ref{fig:rr} 
because $e.meta.\hrel$ and $e.meta.\trel$ 
equal $\hrel$ and $\trel$, respectively. 
Therefore, we create a new compressed edge, 
where $\pprec = e.\pprec \bigoplus e^\prime.\pprec$, 
$\pdep = e.\pdep \bigoplus e^\prime.\pdep$, and $meta = e.meta$. 
If $e$ is an uncompressed edge (a single dependency), 
we compare $\frel(e^\prime)$ and $\frel(e)$ to check whether the two 
edges can be compressed using the \rr pattern.

Next, consider $\ffinddepupdate(e, r)$, 
which finds the \ldeps (denoted as a range $d$) of 
a range of cells $r$ that are contained in $e.\pprec$. 
To determine $d$, we need to 
find its head cell $d_h$ and tail cell $d_t$. 
Since the \sprec of each cell in $d$ forms a sliding window on $e.\pprec$ 
as shown in \pFigure~\ref{fig:rr_find_dep}, the intuition for computing $d_h$ 
is that the top row of $r$ must intersect with the bottom row of $d_h$'s \lprec. 
Similarly, the bottom row of $r$ must intersect 
with the top row of $d_t$'s \lprec. 
So we ``back calculate'' $d_h$ and $d_t$ based on $r$. 
Specifically, we use the following invariant to compute $d_h$, 
\[
d_h + \trel = \pprec_t^{d_h}
\]
where $\pprec_t^{d_h}$ is $d_h$'s \sprec's tail cell 
and $\trel$ is the relative position of $d_h$ with respect to $\pprec_t^{d_h}$ 
as shown in \pFigure~\ref{fig:rr_find_dep}. 
Since $\trel$ is known, the remaining task is to compute $\pprec_t^{d_h}$. 
We know that $\pprec_t^{d_h}$ is in the bottom row of $d_h$'s \lprec 
since it is a tail cell 
and that the bottom row of $d_h$'s \lprec intersects the top row of $r$. 
So $\pprec_t^{d_h}$ is in the top row of $r$ 
and its row index is the row index of $r$'s head cell (i.e., $r.\phead.j$).
Since $\pprec_t^{d_h}$ is a tail cell, 
it is in the right-most column of $e.\pprec$. 
So its column index is $e.\pprec.\ptail.i$.

\begin{figure}[!t]
  \begin{tabular}{cc}
      \begin{minipage}[b]{0.56\linewidth}
        \centering
        \includegraphics[width=0.97\linewidth]{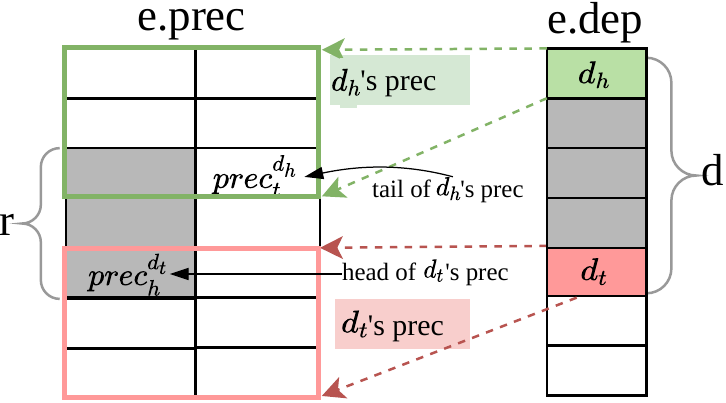}
        \caption{\small An example of $\ffinddepupdate(e, r)$ for \rr}
        \label{fig:rr_find_dep}       
      \end{minipage} 
      \begin{minipage}[b]{0.44\linewidth}
        \centering
        \includegraphics[width=0.98\linewidth]{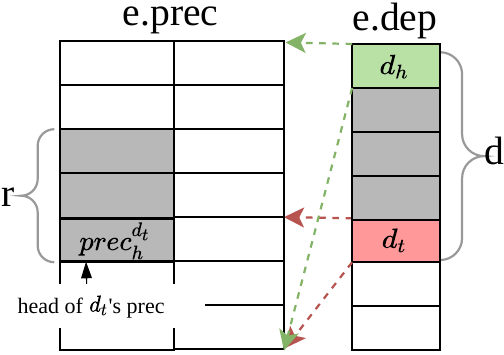}
        \caption{\small An example of $\ffinddepupdate(e, r)$ for \rf}
        \label{fig:rf_find_dep}
      \end{minipage}
  \end{tabular}
  \vspace{-7mm}
\end{figure}

Finding the tail cell $d_t$ adopts a dual procedure. 
Based on the invariant $d_t + \hrel = \pprec_h^{d_t}$, 
we need to find $d_t$'s \lprec's head cell $\pprec_h^{d_t}$. 
As shown in \pFigure~\ref{fig:rr_find_dep}, 
$\pprec_h^{d_t}$ should be in the last row of $r$ 
and in the left-most column of $e.\pprec$. 
Therefore, we have $\pprec_h^{d_t}.i=e.\pprec.\phead.i$ 
and $\pprec_h^{d_t}.j=r.\ptail.j$. 
We note this procedure can output a range $d$ 
that is beyond $e.\pdep$. 
In this case, we take the intersection between $d$ 
and $e.\pdep$ to return a valid range.

The third function, $\ffindprecupdate(e, s)$ finds the \lprecs (denoted as a range $g$)
for a range of cells $s$ contained in $e.\pdep$. 
By the definition of \rr, the \lprecs of the cells in $s$ 
form ``sliding windows'' on $e.\pprec$ as we move from 
$s.\phead$ to $s.\ptail$; $g$ is simply the union of 
the \lprecs of all cells in $s$. 
So $g.\phead$ is the head cell of $s.\phead$'s \lprec 
and $g.\ptail$ is the tail cell of $s.\ptail$'s \lprec. 
We have $g.\phead = s.\phead + \hrel$ and 
$g.\ptail = s.\ptail + \trel$. 

Finally, consider $\fclear(e, s)$, which removes the dependencies 
for a range of formula cells $s$ in $e.\pdep$. 
We first subtract $s$ from $e.\pdep$, 
after which we are left with a range or a union of two ranges. 
For example, if we remove \code{C2} from \code{C1:C4}, 
the remainder is composed of two ranges: \code{C1} and \code{C3:C4}. 
For each range $\varname{newDep}$ in the remaining \sdeps, 
we generate its corresponding \lprec $\varname{newPrec}$ 
using $\ffindprecupdate(e, \varname{newDep})$.

\stitle{\rf, \fr, \ff}
We now discuss the key functions for \rf. 
The function $\fcompress(e, e^\prime)$ for the \rf pattern
has similar logic to the one for \rr 
and is different only 
in the compression condition. 
By definition of \rf, we first compute the relative position 
between $e^\prime.\pdep$ and the head cell of $e^\prime.\pprec$ (denoted $\hrel$) 
and check whether $\hrel$ and $e.meta.\hrel$ are the same.
If so, we additionally check whether 
the tail cell of $e^\prime.\pprec$ is the same as $e.meta.\tfix$. 

$\ffinddepupdate(e, r)$ finds the range of dependents $d$ 
for a range of cells $r$ contained in $e.\pprec$. 
Similar to \rr, we need to find $d$'s head cell $d_h$ and tail cell $d_t$. 
To compute $d_h$, we use the intuition shown in \pFigure~\ref{fig:rf_find_dep}:
the \lprec of $e.\pdep.\phead$ equals $e.\pprec$ and the \lprec of each cell 
in $e.\pdep$ shrinks when we move from $e.\pdep$'s head to its tail cell. 
That is, $e.\pdep.\phead$ references the entire range of $e.\pprec$ 
and is the dependent of any $r$ contained in $e.\pprec$. 
So $e.\pdep.\phead$ equals $d_h$.

To compute $d_t$, we use the observation that the \lprec of each cell 
in $e.\pdep$ shrinks as we move from $d_h$ to $d_t$ 
such that the bottom row of $r$ should intersect 
with the top row of the \lprec of $d_t$.
Therefore, to compute $d_t$, we leverage the invariant 
$d_t + \hrel = \pprec_h^{d_t}$, where $\pprec_h^{d_t}$ is 
the head cell of $d_t$'s \lprec. 
Since $\hrel$ is known, we need to compute $\pprec_h^{d_t}$. 
As \pFigure~\ref{fig:rf_find_dep} shows, 
since $\pprec_h^{d_t}$ is in the bottom row of $r$, 
its row index is $r.\ptail.j$ 
and since $\pprec_h^{d_t}$ is a head cell, 
its column index is $e.\pprec.\phead.i$. 

Consider $\ffindprecupdate(e, s)$, 
which finds the \lprecs (denoted as a range $g$) of $s$ contained in $e.\pdep$. 
Our observation is that the precedent of $s.\phead$ contains 
all of the precedents of other cells in $s$ 
since the \lprec of each cell in $e.\pdep$ is shrinking 
as we move from $s.\phead$ to $s.\ptail$.
Therefore, $g$ is the \lprec of $s.\phead$, 
and is computed as $g.\phead = s.\phead + \hrel$ and $g.\ptail = \tfix$. 
The function $\fclear(e, s)$ for \rf follows the same logic as \rr. 
We first remove $s$ from $e.\pdep$ to return one or two ranges. 
For each returned range $\varname{newDep}$,
we generate their corresponding \lprec $\varname{newPrec}$ 
using $\ffindprecupdate(e, \varname{newDep})$ of \rf.  

\fr is a dual pattern of \rf, 
so its algorithms can be easily derived from the algorithms above. 
\ff's algorithms can also be derived from \fr and \rf. 
We omit these algorithms due to space limits. 

\stitle{Algorithmic complexity}
All algorithms for the basic patterns 
are $O(1)$, independent  of the number of edges compressed.

\comment{
\subsection{Pattern Space} 
\label{sec:pattern_space}
The basic patterns can be mixed to express complex patterns 
where formula cells reference multiple ranges. 
Consider the example in \pFigure~\ref{fig:locality_example}. 
We will have four compressed edges for this example, 
where $\code{A$i$} \rightarrow \code{N$i$}$, 
$\code{A($i$-1)} \rightarrow \code{N$i$}$, 
$\code{M$i$} \rightarrow \code{N$i$}$, 
and $\code{N($i$-1)} \rightarrow \code{N$i$}$ 
can be compressed using \rr. 
We can also handle messy spreadsheets, where the formula cells 
do not strictly follow the defined patterns. 
For example, if \code{N4=IF(A4=A3,N3+M4,M4)} in 
\pFigure~\ref{fig:locality_example} is modified to be 
\code{N4=IF(A4=A3,N3+M4,M4\textcolor{red}{+K5})}. 
We should still extract the four compressed edges and regard 
$\code{K5} \rightarrow \code{N4}$ as an uncompressed edge. 
In \pSection~\ref{sec:sys}, we will discuss a 
general compression framework 
that extracts basic patterns from 
complex and messy formula references. 
Beyond the basic patterns, 
we additionally consider formula cells 
that are away from each other at a fixed gap along a column or row 
(e.g., every other row in a column) 
and a special pattern of \rr to accelerate finding \sdeps in \rr. 
These new patterns appear often in practice and are discussed 
in \pSection~\ref{sec:extension}. 
}

\section{\sys Framework}
\label{sec:sys}

\rtwo{We now introduce the \sys framework, 
which includes four generic and extensible algorithms 
for efficiently compressing, querying, 
and modifying formula graphs. 
These algorithms are extensible since 
they only utilize the four key functions per pattern 
as shown in Section~\ref{sec:keyfunctions} 
and any pattern can be integrated into the \sys framework 
if they implement these functions.} 
In this section, we first introduce these algorithms 
\ptr{and analyze their complexity}
(\pSection~\ref{sec:compress}-\ref{sec:maintain}).
Then, we compare the complexity of \sys against an approach that 
does not compress the formula graph (\pSection~\ref{sec:complexity}).
\ptr{We assume the formula graphs 
in both approaches are implemented via an adjacency list 
and we build an R-Tree index on the vertices to quickly 
find the overlapping vertices for an input range.}
\ppaper{For both approaches, we build an R-Tree index 
on the vertices to quickly find the overlapping vertices 
for an input range.}
The complexity of operations on an R-Tree varies
based on the design choices. 
Our analysis assumes the complexity for searching, 
inserting, and deleting one range is $O(N)$, $O(\log N)$, 
and $O(\log N)$, respectively, 
where $N$ is the number of ranges stored in the R-Tree. 

\ptr{
\begin{algorithm}[!t]
  \small
  \SetAlgoLined\DontPrintSemicolon
  \SetArgSty{textrm}
  \SetKw{break}{break}
  \SetKwFunction{algo}{\fcompressname}
  \SetKwFunction{proc}{genCompEdges}
  \SetKwProg{myalg}{Algorithm}{}{}
  
  \myalg{\algo{$G(E, V)$, $e^\prime$}}{
      $\varname{isCompressed} \gets$ \text{false}\;

        $\varname{pSet} \gets$ pre-defined patterns\;
        $\varname{eSet} \gets$ \text{find all} $e \in E$ \text{whose} $e.\pdep$ \text{is} \text{adjacent to} \;
        \text{\hspace{10mm}} $e^\prime.\pdep$ \text{on column or row axis} \;
        
        \For{$\varname{candE} \in \varname{eSet}$}{
            $\varname{edgePairs} \gets$ \proc{$\varname{candE}$, $e^\prime$, $\varname{pSet}$}\;
            $\varname{edgePairSet}.\procname{add}(\varname{edgePairs})$\;
        }
        
        \If{$\varname{edgePairSet}$ \text{is not empty}}{
           $\varname{edgePair} \gets$ \text{sort} $\varname{edgePairSet}$ \text{by heuristics and}\;
           \text{\hspace{17mm}take the first}\;
           \text{maintain} $G$ \text{using} $\varname{edgePair}$\;
           $\varname{isCompressed} \gets$ \text{true}\;
           \break\;
        }
      \uIf{$\varname{isCompressed}$ \text{is false}}{
        \text{insert} $e^\prime$ \text{into} $G$\;
      }
  }{}
  
  \SetKwProg{myproc}{Procedure}{}{}
  \myproc{\proc{$\varname{candE}$, $e^\prime$, $\varname{pSet}$}}{
    \uIf{$\varname{candE}.p == \ncPattern$}{
        \For{$p \in \varname{pSet}$}{
            $\varname{pair} \gets (p.\fcompress(\varname{candE}, e^\prime), \varname{candE})$ \;
            $\varname{edgePairs}.\procname{addIfValid}(\varname{pair})$\;
        }
    } \Else {
        $\varname{pair} \gets (\varname{candE}.p.\fcompress(\varname{candE}, e^\prime), \varname{candE})$\;
        $\varname{edgePairs}.\procname{addIfValid}(\varname{pair})$\;
    }
    \KwRet $\varname{edgePairs}$ \;
  }
  \caption{\small Compressing a dependency $e^\prime$ into $G(E, V)$}
  \label{alg:compression}
\end{algorithm}}

\begin{figure*}[!t]
\vspace{-20pt}
    \centering
    \includegraphics[height=27mm]{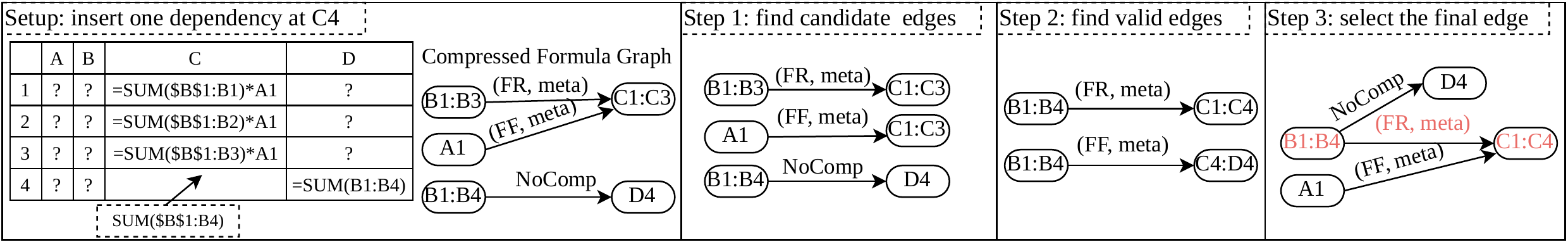}
    \vspace{-1mm}
    \caption{An example of compressing one dependency inserted at \code{C4}}
    \label{fig:compression_example}
    \vspace{-5mm}
\end{figure*}

\subsection{Compressing a Formula Graph}
\label{sec:compress} 

We formalize the problem of 
minimizing the number of edges of the compressed formula graph 
based on the predefined basic patterns 
and present our compression algorithm.

\subsubsection{Problem formalization}
Using the definition of the compressed graph $G$, 
the problem of minimizing the number of edges in $G$ 
is equivalent to the problem of 
finding a partition $P = \{E_1^\prime, E_2^\prime, \cdots, E_N^\prime\}$ 
of the uncompressed edge set $E^\prime$ such that $N$ is minimum, 
and each $E_i^\prime$ is compressed by a single pattern or 
only includes one uncompressed edge. 
\ppaper{We call this problem Compressed Edge Minimization (or CEM for short).}
\ptr{The optimization problem, which we call Compressed Edge Minimization (or CEM for short), is defined as follows: 
\begin{equation*}
    \small
    \begin{aligned}
    & \underset{P = \{E_1^\prime, \cdots, E_N^\prime\}}{\text{minimize}} & & N \\
    & \text{where} & & E_i^\prime \text{ is compressed by a pattern or is} 
    \\ & & & \text{an uncompressed edge,} \forall i \in 1\ldots N \\
    & & & \, \cup_i E_i^\prime = E^\prime
    \end{aligned}
\vspace{-1mm}
\end{equation*}}
We now show CEM is {\sc NP-Hard} 
\ptr{even when we only consider each basic pattern}. 
\begin{theorem}[CEM Hardness]
\vspace{-1mm}
\rone{Compressing the graph $G^\prime$ into $G$
while minimizing the number of edges of $G$ is {\sc NP-Hard}, 
even when restricted to each basic pattern.}
\vspace{-1mm}
\end{theorem}
\begin{proof}
(Sketch) 
\rone{
We reduce 
the \textit{rectilinear picture compression} (RPC for short) 
problem, which is known to be {\sc NP-Hard}~\cite{intractability}, to CEM. 
The input to RPC is a $m\times n$ matrix of 0's and 1's,
with the goal to find the minimal number of rectangles 
that precisely cover the 1's. 
We reduce RPC to CEM by mapping 
the input matrix to a spreadsheet range $R$: 
for each value at column $i$ and row $j$ in the matrix, 
if the value is 1, we place a formula at $(i, j)$ in $R$; 
otherwise, we place a pure value at $(i, j)$ in $R$. 
If the dependencies of any range of formulae in $R$ 
can be compressed into a single edge by a pattern $p$, then 
the RPC problem is equivalent to minimizing the number of 
edges of the compressed graph $G$ for $R$ if we only consider 
the pattern $p$. }

\rone{
Therefore, we need to construct a range $R$ for each pattern $p$ 
such that the dependencies of any range of formulae in $R$ 
can be compressed into a single edge by $p$. 
For \ff, we let all of the formulae in $R$ 
reference the same fixed range outside $R$ to 
meet the above condition. 
Similarly, for \rr, we let each formula in $R$ 
reference the cell to its left. 
For \rf, we first construct another range $R^\prime$ 
that can be compressed into a single edge by \rf 
and has the same shape as $R$. 
Note that in $R^\prime$, the dependencies of any range of formulae 
can also be compressed into a single edge by \rf. 
Afterwards, we generate $R$ by 
copying the formula at $(i, j)$ in $R^\prime$ 
to $(i, j)$ in $R$ if the cell $(i, j)$ in $R$ is a formula cell. 
This way, we construct a range $R$ where the dependencies 
of any range of formulae in $R$ can be compressed into a single 
edge by \rf. The case for \fr is done symmetrically.}
\end{proof}
\ptr{CEM is also trivially {\sc NP-Complete} since verifying
that a partition using \ff is correct is in {\sc PTIME}.}
We tested the algorithm that enumerates all possible partitions 
and found it cannot finish within 30 mins for a spreadsheet with 
96 edges because the number of possible partitions 
is a Bell number~\cite{bellnumber}. 
To reduce the compression overhead, 
we propose a greedy compression algorithm.  


\subsubsection{Greedy compression algorithm}
Our algorithm compresses a list of dependencies
between formula cells and their referenced ranges 
by repeatedly inserting each dependency into the compressed graph 
and determining the partitions as well 
as the corresponding compression patterns. 
Observing that pre-defined patterns compress 
the dependencies in adjacent formula cells, 
we use this constraint to quickly find the candidate edges 
that one inserted dependency can be compressed into.
If there are multiple candidate edges, 
we leverage several heuristics based 
on our analysis of real-world spreadsheets
to decide the edge that can best reduce the graph size 
(e.g., by leveraging the dollar sign cues in the 
formula expression if available).

\ptr{Algorithm~\ref{alg:compression} shows our approach for 
compressing one dependency $e^\prime = (\pprec, \pdep)$ 
into a compressed formula graph $G$.}
\ppaper{The algorithm compresses one dependency 
$e^\prime = (\pprec, \pdep)$ into a compressed formula graph $G$}
We use the example in \pFigure~\ref{fig:compression_example} 
to explain compressing $e^\prime$ into $G$.
The setup of \pFigure~\ref{fig:compression_example} (the left pane) shows 
each formula cell in column \code{C} referencing two ranges. 
The references to column \code{B} follow \fr  
and the references to column \code{A} follow \ff. 
In addition, we have an uncompressed edge of 
\code{D4} referencing \code{B1:B4}. 
Our example assumes that a formula \code{SUM(\$B\$1:B4)} 
is inserted at \code{C4} (i.e., $e^\prime=(\code{B1:B4},\code{C4})$). 

\stitle{Find candidate edges:}
The first step is to quickly find candidate edges that the dependency $e^\prime$ can be compressed into. 
Specifically, an edge $e$ is a candidate if 
$e.\pdep$ is adjacent to $e^\prime.\pdep$ 
along the row or column axis.
Step 1 in \pFigure~\ref{fig:compression_example} 
shows that all three edges  
meet this condition because $e^\prime.\pdep = \code{C4}$ is adjacent to 
both \code{C1:C3} and \code{D4}. 
To find these edges, 
we first shift $e^\prime.\pdep$ by one cell 
in all four directions (i.e., up/down/left/right)
and use the index on the vertices (e.g., an R-Tree~\cite{r-tree})
to quickly find ranges that overlap with the shifted $e^\prime.\pdep$. 
Then, for each overlapping range (e.g., \code{D4}), 
we find its \sprec (e.g., $\code{B1:B4}$) 
and add this edge (e.g., $\code{B1:B4} \rightarrow \code{D4}$) 
into the candidate edge set.

\stitle{Find valid candidates:} 
Next, we check whether $e^\prime$ can be compressed into 
each candidate edge using $\fcompress(e, e^\prime)$ 
from \pSection~\ref{sec:locality} to find the valid compressed edges
\ptr{ (i.e., \procname{genCompEdges} in Algorithm~\ref{alg:compression})}. 
We consider two cases. 
First, if the candidate edge $\varname{candE}$ is not compressed, 
we check whether $e^\prime$ and $\varname{candE}$ 
can be compressed into a new edge $\varname{newEdge}$ 
using the predefined patterns. 
If so, we store $\varname{newEdge}$ as a valid candidate edge
\ptr{ (i.e., \procname{addIfValid} in Algorithem~\ref{alg:compression})}. 
If, instead, $\varname{candE}$ is a compressed edge, 
we check whether $e^\prime$ can be compressed into $\varname{candE}$ 
and if so, we generate a valid edge. 
Step 2 in \pFigure~\ref{fig:compression_example} shows 
two valid compressed edges because 
the edge $\code{B1:B4} \rightarrow \code{C4}$ can be compressed into 
$\code{B1:B3} \rightarrow \code{C1:C3}$ 
or $\code{B1:B4} \rightarrow \code{D4}$. 

\stitle{Select the final edge:} 
The final step is to select the final edge from the valid ones. 
The selection is based on the following heuristics, in order. 
First, we prioritize column-wise compression over row-wise compression. 
If this heuristic does not return a single edge, 
we further compare the priority of each remaining edge's pattern. 
If one pattern $p_a$ is a special case of another pattern $p_b$, 
then we choose $p_a$ over $p_b$ because 
we expect the special pattern $p_a$ 
to be more efficient. In \pSection~\ref{sec:extension}, 
we will describe one such special pattern of $\rr$. 
Otherwise, we leverage the dollar sign (\$) information, 
if available, sometimes specified as part of the formula strings. 
For example, for the formula string $\code{SUM(\$B\$1:B4)}$ 
at $\code{C4}$, we will prioritize compressing its dependency 
$\code{B1:B4} \rightarrow \code{C4}$ using $\fr$ over other patterns 
because the head cell \code{\$B\$1} in $\code{SUM(\$B\$1:B4)}$ 
has dollar sign annotations, but its tail cell does not, 
which indicates that $\code{SUM(\$B\$1:B4)}$  
follows the $\fr$ pattern if it is generated via autofill. 
For our example in \pFigure~\ref{fig:compression_example}, 
we choose the compressed edge 
($\code{B1:B4} \rightarrow \code{C1:C4}$) 
over ($\code{B1:B4} \rightarrow \code{C4:D4}$) 
because the former one uses column-wise compression. 
Finally, we delete the old edge and insert the newly compressed edge. 

\ptr{
\stitle{Algorithmic complexity}
Our analysis assumes the inserted dependencies 
have no duplicates, 
so its size is $|E^\prime|$, the number of uncompressed edges. 
For each inserted dependency, we leverage the R-Tree to find
the candidate edges, taking $O(|V|)$ operations, 
where $|V|$ is the number of vertices 
in the compressed formula graph $G$. 
The number of the candidate edges is $O(|E|)$. 
For these candidate edges, 
it takes $O(|E|)$ operations to find the valid compressed edges 
and the final edge that the input dependency 
is compressed into. In addition, we need to maintain the R-Tree 
by removing the old and inserting the new vertices, 
which takes $O(\log |V|)$ operations. 
In total, the complexity of inserting $|E^\prime|$ dependencies 
is $O(|E^\prime| \times (|V| + |E| + \log |V|))$ = 
$O(|E^\prime| \times |E|)$ since each vertex 
is connected to at least one edge. }

\begin{algorithm}[t]
    \small
	\SetAlgoLined\DontPrintSemicolon
    \SetArgSty{textrm}
    
    initiate $\varname{queue}$ as a queue containing only $r$\;
    initiate $\varname{result}$ as an empty set and an R-Tree for it \; 
    \While{$\varname{queue}$ is not empty} {
        $\varname{\precU} \gets$ \text{remove the first element in} $\varname{queue}$\;
        $\varname{precs} \gets$ \text{find vertices that overlap with} \;
        \text{\hspace{10mm}} $\varname{\precU}$ \text{via the R-Tree on} V\;
        \For{$\varname{prec} \in \varname{precs}$}{
            $\varname{edges} \gets$ $\{e: e\in E$ \text{and} $e.\pprec = \varname{prec}\}$\;
            \For{$e \in \varname{edges}$}{
                $\varname{dep} \gets e.p.\ffinddepupdate(e, \varname{\precU})$\;
                $\varname{newDepSet} \gets$ $\text{Find the subset of } \varname{dep} \text{ not}$ \;
                $\text{\hspace{7mm}contained in } \varname{result} \text{ via the R-Tree on } \varname{result}$\;
                \For{$\varname{newDep} \in \varname{newDepSet}$}{
                    \text{add} $\varname{newDep}$ \text{to} $\varname{result}$ \text{and its R-Tree}\;
                    \text{add} $\varname{newDep}$ \text{to} $\varname{queue}$\;
                }
            }
        }
    }
    \Return $\varname{result}$

    \caption{\small Find \sdeps of a column/row of cells $r$ in $G(E, V)$}
	\label{alg:query}
\end{algorithm}

\subsection{Querying a Formula Graph}
\label{sec:query} 
We now discuss finding the \sdeps or \sprecs of 
a column/row of cells $r$ 
in $G$ using the key functions from \pSection~\ref{sec:locality}. 
Since finding \ldeps is the dual problem of finding \lprecs, 
we focus on the former. 
We apply Breadth-First-Search (BFS), but with three major differences. 
First, when we find the direct \ldeps of $r$, 
we need to consider all of the vertices in $G$ that overlap with $r$. 
Second, since an edge $e$ in $G$ can be a compressed edge, 
finding the direct \ldeps of $r$ in $e$ may not be the full $e.\pdep$, 
but a subset instead. So we need to 
find the real \ldeps within $e.\pdep$. 
Third, for a ``real'' \ldep, which will in turn serve as a \sprec 
for subsequent searches, we need to add the subset of this \ldep 
that has not yet been visited during BFS.
We illustrate the three modifications below. 

Algorithm~\ref{alg:query} shows the modified BFS algorithm: 
it takes a column or row of cells $r$ as input and 
returns the set of ranges that depend on $r$. 
We explain this algorithm using the compressed graph in Step 3 
of \pFigure~\ref{fig:compression_example}. 
Our example involves finding the \ldeps of \code{B2}. 
Our algorithm uses a queue to store the ranges to be visited 
in the future and a set $\varname{result}$ to store 
the ranges that depend on $r$ and have been visited. 
An additional R-Tree is also built for $\varname{result}$.
For each range $\varname{\precU}$ in this queue,
we find its direct \ldeps. 
As mentioned earlier, we need to consider the ranges that overlap 
with $\varname{\precU}$ (i.e., \code{B1:B4} for \code{B2}). 
Next, we find the direct \ldeps of each overlapping range 
and the corresponding edges 
(i.e., $\code{B1:B4} \rightarrow \code{C1:C4}$ 
and $\code{B1:B4} \rightarrow \code{D4}$). 
Since some of these edges can be compressed, 
for a compressed edge $e$ 
we need to find the real direct \ldep within $e.\pdep$ 
for $\varname{\precU}$, 
which is done by the key function 
$\ffinddepupdate(e, \varname{\precU})$. 
For the example of the input \code{B2}, 
we return \code{C2:C4} for the edge 
$\code{B1:B4} \rightarrow \code{C1:C4}$ 
since \code{C1} does not depend on \code{B2}. 
Finally, we find the subset of the real \ldep 
that has not yet been visited via the R-Tree 
on the \varname{result} set, 
and add the subset to the queue, the set \varname{result}, 
and the R-Tree on \varname{result}. 
We repeat the process until the queue is empty. 
For the example of the \ldep \code{C2:C4}, 
if we have visited \code{C2:C3}, 
we will only store \code{C4} in the queue 
and \varname{result}.

\ptr{
\stitle{Algorithmic complexity}
To analyze the complexity of Algorithm~\ref{alg:query}, we consider 
two cases: 1) Algorithm~\ref{alg:query} accesses each edge in $G$ 
at most once; 2) otherwise. 
For the first case, each vertex in $G$ will serve 
as the \sprec at most once when we find direct \sdeps 
(i.e., the inner part of the first \textbf{for} loop in Algorithm~\ref{alg:query}). 
So it will only repeat $O(|V|)$ times. 
To find a \sprec (i.e., \varname{prec} in 
the first \textbf{for} loop in Algorithm~\ref{alg:query}), 
we need to search the R-Tree, taking $O(|V|)$. 
To find the real direct \sdeps for a \sprec using $\ffinddepupdate$, 
we need to spend $O(dep\_num)$ operations, where $dep\_num$ 
is the number of direct dependents of a \sprec 
and $\ffinddepupdate$ takes a constant time. 
For each real direct \sdep, we need to additionally 
find the subset that is not contained in \varname{result} 
using the R-Tree on \varname{result}, 
taking $O(\text{size of \varname{result}})$. 
Since each range in \varname{result} 
will be a \sprec, the size \varname{result} is $O(|V|)$. 
In addition, the total cost for maintaining 
the R-Tree on \varname{result} is $O(|V| \times \log|V|)$ 
since the size of \varname{result} is $O(|V|)$. 
To sum up, the complexity here 
is $O(|V|) \times (O(|V| + O(dep\_num) \times O(|V|)) 
+ O(|V| \times \log|V|)= O(|V|^2 + |E|\times|V|)$. 

For the second case, 
the first \textbf{while} loop in Algorithm~\ref{alg:query}
runs $O|V^\prime|$ times, where $|V^\prime|$ 
is the number of vertices in the uncompressed graph. 
This is because the size of \varname{result} is $O(|V^\prime|)$, 
so the total number of ranges inserted 
into the queue is also $O(|V^\prime|)$.  
For each range $\varname{\precU}$ in the queue, 
we check $O(|E|)$ edges in $G$. 
For each edge, we find the real dependent using $O(1)$ 
and take $O(|V^\prime|)$ operations 
to find the subset of the real dependent that 
is not contained in \varname{result}. 
The total cost for maintaining 
the R-Tree on \varname{result} is $O(|V^\prime| \times \log|V^\prime|)$. 
So the complexity is 
$O(|V^\prime|) \times (O(|V|) + O(|E|\times|V^\prime|)) + 
O(|V^\prime| \times \log|V^\prime|$ = 
$O(|V^\prime|^2\times|E|)$. }


%
%

\subsection{Maintaining a Formula Graph}
\label{sec:maintain}
We now discuss maintaining the formula graph 
when users insert, clear, or update formula cells. 
We process inserts using \ptr{Algorithm~\ref{alg:compression}}
\ppaper{the algorithm from Sec.~\ref{sec:compress}}. 
Since an update can be modeled as a clearing operation plus an insert, 
we focus on clearing formula cells. 

The idea of clearing a column/row of formula cells $s$ 
is to delete $s$ from the edges whose \ldeps (i.e., formula cells) 
overlap with $s$ using $\fclear(e,s)$ from \pSection~\ref{sec:locality}. 
First, we find the relevant edges $\varname{relEdges}$ 
whose \ldeps overlap with $s$. 
Second, for each $e \in \varname{relEdges}$, 
we generate new edges $\varname{newESet}$ 
after clearing $s$ for $e$, which is done by $\fclear(e,s)$.  
Finally, we maintain the graph 
by deleting the old edge $\varname{relEdges}$ 
and inserting the new edges $\varname{newESet}$. 

\ptr{
\stitle{Algorithmic complexity}
For this algorithm, the size of the relevant edges whose 
dependents overlap with \varname{s} is $O(|E|)$ 
and searching the R-tree takes $O|V|)$, 
so the cost for finding relevant edges is $O(|E|)$. 
From each relevant edge, clearing \varname{s} and 
maintaining the graph is $O(1)$ 
while maintaining the R-tree takes $O(\log |V|)$ time. 
In total, the complexity for removing \varname{s} 
is $O(|E| \log |V|)$.}


\begin{table}[!t]
\begin{tabular}{lll}
\hline & TACO  & NoComp \\ \hline
Building    &  $O(|E^\prime| \times |E|)$ & $O(|E^\prime| \times \log |V^\prime|)$   \\
Querying    & \begin{tabular}[l]{@{}l@{}} Case 1: $O(|V|^2 + |E|\times|V|)$ \\ Case 2: $O(|V^\prime|^2\times|E|)$ \end{tabular} & $O(|V^\prime|^2 + |E^\prime|)$      \\
Maintaining & $O(|E| \log |V|)$ & $O(|E^\prime| \log |V^\prime|)$ \\ \hline
\end{tabular}
\caption{\small Complexity comparison between \sys and \nocomp}
\label{tbl:complexity}
\vspace{-7mm}
\end{table}

\subsection{Complexity Comparison with a No Compression Approach}
\label{sec:complexity}

We now compare the complexity of \sys 
with an approach that does not compress the formula graph, 
called \nocomp. \ppaper{Similar to \sys, 
it adopts an adjacency list, 
builds an R-Tree on its vertices, uses BFS to find \sdeps. 
Recall that the uncompressed formula graph is represented 
as $G^\prime = (E^\prime, V^\prime)$.} 
\ppaper{Due to space limit, details about \sys and \nocomp's 
complexity are in our technical report~\cite{taco_tr}.}
Table~\ref{tbl:complexity} summarizes the results.
\ppaper{Note that to analyze the complexity of 
the algorithm for finding \sdeps (i.e., Algorithm~\ref{alg:query}), 
we consider two cases: 1) Algorithm~\ref{alg:query} 
accesses each edge in $G$ at most once (denoted as Case~1); 
2) otherwise (denoted as Case~2).}

\ptr{
\nocomp builds the uncompressed formula graph $G^\prime$ 
by inserting a list of dependencies into $G^\prime$. 
For each dependency, we need to insert its precedent 
and dependent into an R-Tree, 
taking $O(\log |V^\prime|)$ operations, and insert the dependency 
into the adjacency list, taking $O(1)$. 
In total, the complexity of inserting $|E^\prime|$ dependencies 
is $O(|E^\prime| \times \log |V^\prime|)$. 
\sys can be more expensive than \nocomp for building the formula graph 
because it needs to search the R-Tree and find the edge 
that an inserted dependency can be compressed into.

Next, we analyze the complexity of finding \sdeps of 
an input range $r$ via a modified BFS. 
During BFS, when we find the direct \ldeps of an input range $r$, 
we need to consider all of the vertices in $G^\prime$ 
that overlap with $r$ (i.e., via an R-Tree search). 
Similar to conventional BFS, it recursively finds 
\sdeps starting from the input range $r$. 
Each vertex in $G^\prime$ may serve as a \sprec 
when we find direct \sdeps of a \sprec (i.e., $O(|V^\prime|)$ times), 
while the cost for finding one \sprec via the R-Tree is $O(|V^\prime|)$. 
Finding direct \sdeps of one \sprec is $O(dep\_num)$, 
where $dep\_num$ is the number of the direct \sdeps. 
The overall complexity for finding \sdeps is 
$O(|V|^\prime) \times (O(|V^\prime|) + O(dep\_num)) 
= O(|V^\prime|^2 + |E^\prime|)$.}

\ppaper{We see that \sys can be more expensive 
than \nocomp for building the formula graph 
due to its compression overhead, 
but \sys is more efficient for modifying the formula graph as 
shown in Table~\ref{tbl:complexity}.} 
\ppaper{For finding \sdeps,}
\ptr{We see that}\sys is more efficient than \nocomp if 
the querying algorithm of \sys accesses each edge 
in the compressed graph $G$ at most once  
(i.e., Case 1 in Table~\ref{tbl:complexity}).  
For Case 2, \sys can be potentially more expensive than \nocomp in theory. 
To understand the performance of Case 2, 
we analyze real spreadsheets to find when this case will happen 
and become the performance bottleneck, 
and adopt an extended pattern to reduce the cost for this case 
in \pSection~\ref{sec:extension}.
In practice, we find the average number of accesses 
for an edge during BFS is relatively low. 
For the tests for finding \sdeps in \pSection~\ref{sec:exp}, 
the average number of edge accesses during BFS is 
no larger than 7 for 98\% of the tests. 
In addition, our experiments in \pSection~\ref{sec:exp} 
show that \sys is much more efficient than \nocomp  
on real spreadsheets.

\ptr{
Finally, clearing a column/row of formula cells \varname{s} 
requires searching the R-tree to 
find relevant edges (i.e., $O(|E^\prime|)$). 
Then, we will delete each relevant edge and update the 
R-tree (i.e., $O(\log |V^\prime|)$). 
In total, the complexity is $O(|E^\prime| \log |V^\prime|)$.
As shown in Table~\ref{tbl:complexity}, 
\sys is more efficient here. }

\section{Extension and Limitations}
\label{sec:extension}

%




\begin{figure}[!t]
  \centering
  \includegraphics[width=0.8\linewidth]{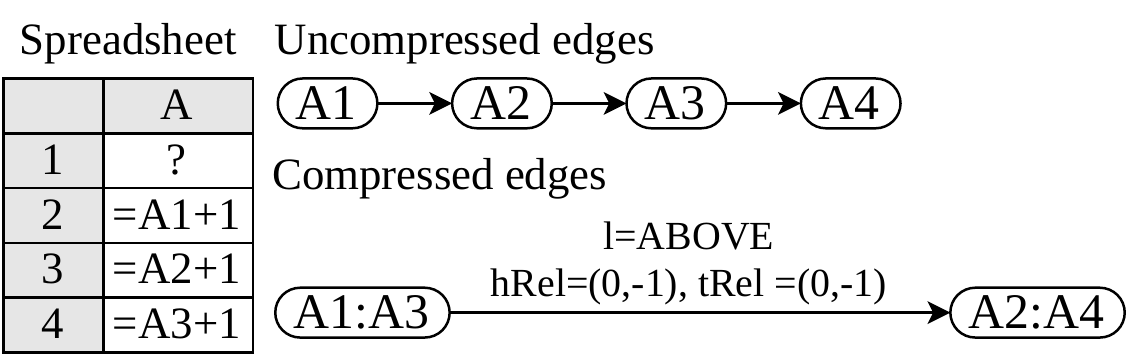}
  \vspace{-2mm}
  \caption{One example of \rrChain}
  \label{fig:rrChain}
  \vspace{-7mm}
\end{figure}

\stitle{Supporting a new pattern: \rrChain} 
As shown in \pSection~\ref{sec:sys}, 
our algorithm for finding \sdeps or \sprecs may be slow 
if an edge is repeatedly accessed multiple times. 
By examining real spreadsheets, we find one pattern 
that leads to these cases and 
becomes a performance bottleneck for \sys. 
In this section, we discuss supporting this pattern to 
further accelerate \sys. 
Our discussion focuses on a column of cells and finding \sdeps
as before; the other cases can be derived symmetrically.

Consider a column of formula cells that 
form a chain of dependencies, 
where each formula cell references its adjacent 
formula cell above or below. 
We will compress these dependencies using \rr 
because each formula cell has the same relative position 
with respect to its referenced range. 
Consider the example in \pFigure~\ref{fig:rrChain}, 
where each formula cell starting from \code{A2} 
increments the value of the above formula cell by one. 
To find the \sdep{s} of \code{A1}, 
we first find its overlapping vertex \code{A1:A3} 
and then compute its real direct \sdep: {A2}. 
Afterwards, the compressed edge is repeatedly accessed 
until we reach the end of this chain, 
which introduces high searching overhead.

To solve this problem,
we introduce a new pattern \rrChain as a special case of \rr. 
\rrChain's $\pmeta$ additionally includes a variable $l$ to 
indicate the direction of a formula cell referencing its adjacent cell. 
For example, $l$ is \code{ABOVE} in \pFigure~\ref{fig:rrChain} 
because each formula cell references its adjacent cell above. 
Our discussion focuses on $l=\code{ABOVE}$; 
the case for $l=\code{BELOW}$ can be easily derived. 
To compress a dependency $e^\prime$ into $e$ for \rrChain, 
we first check the condition of \rr, 
and then further check whether 
$e^\prime.\pprec$ is above $e^\prime.\pdep$ and 
if they are adjacent. 
To find \sdeps of a range $r$,  
we return a range $d$ between $r.\phead$'s direct \sdep 
and the tail cell of $e.\pdep$. 
Consider finding \sdeps of \code{A2} in \pFigure~\ref{fig:rrChain}. 
We return the range between \code{A3} 
(i.e., \code{A2}'s direct \sdep) and 
the tail cell of $e.\pdep$ (i.e., \code{A4}). 
Finally, clearing formula cells in $e.\pdep$ follows 
the same logic as \rr and is omitted. 

\stitle{Limitations} 
\rtwo{One limitation of the patterns in \sys 
is that they focus on adjacent formula cells. 
So, they only represent a subset  
of tabular locality in spreadsheets. 
It is possible to extend these patterns and exploit 
other patterns to further 
reduce the sizes of formula graphs. 
For example, one extended pattern derived from \rr could be 
that the referenced ranges in the formula cells 
of every other row follow the \rr pattern (denoted as \rrGapOne).  
We have tested its prevalence 
and found it is much less prevalent than its \rr counterpart. 
Specifically, \rrGapOne reduces the number of edges by 195k 
and 275k for Enron and Github datasets, respectively, while 
the number of edged reduced by \rr is 17.4M and 141.9M, respectively 
(Details in Section~\ref{sec:exp_graph_size}).
Also, it is possible to exploit other information, 
such as a column of formulae having the same functions, 
to better compress formula graphs. 
Fully exploiting these patterns and information is left to future work.}

\comment{
\subsection{\rrGap}
Another type of new pattern we discovered is 
that formula cells in a column 
repeat the \rr pattern for every few rows or columns. 
This pattern is often employed  
when a user has a small template with formulae 
(e.g., a set of formulae calculating inventory for a day) 
and repeatedly uses this template across rows/columns.
We illustrate this pattern with gap 1, \rrGapOne. 
An example is shown in \pFigure~\ref{fig:rrGapOne}. 
We can directly use the $\compress(e, e^\prime)$ function 
from \rr to perform compression 
because Algorithm~\ref{alg:compression} guarantees 
the gap size constraint. 
To find the \sdeps of an input range $r$ in a compressed edge $e$, 
we first use the algorithm for \rr to find the output range $d$ 
and then find the valid cells in $d$ that respect the gap $g=1$. 
For example, the compressed edge will return \code{(B3:B5)} 
for the input \code{A3}. But we will only return 
\code{B3} and \code{B5} since \code{B4} does not depend on \code{A3}.
Clearing a range $s$ in $e$ follows a similar logic 
to $\fclear(e, s)$ of \rr. 
The only difference is that when we remove $s$ from $e.\pdep$ and 
generate one or two new ranges, 
each generated range $\varname{newDep}$ needs to be corrected to guarantee 
that both its head and tail cell are valid in $e$. 
For example, if we clear \code{B5} from 
\code{(B3:B5)} in \pFigure~\ref{fig:rrGapOne}, 
$\varname{newDep}$ should be \code{B3} rather than \code{(B3:B4)}. 
}

\section{Experiments}
\label{sec:exp}

\begin{table}
\fontsize{6.50}{8.5}\selectfont
\centering
\begin{tabular}{lcccc}
\hline
\multicolumn{1}{c}{\multirow{2}{*}{}} & \multicolumn{2}{c}{Enron} & \multicolumn{2}{c}{Github} \\ \cline{2-5} 
\multicolumn{1}{c}{}                  & Vertices      & Edges     & Vertices      & Edges      \\ \hline
\nocomp                        & 18.6M         & 23.7M     & 165.8M        & 179.8M     \\
TACO-InRow                            & 7.7M  (41.2\%)     & 12.5M (52.8\%)     & 55.2M (33.3\%)         & 55.2M (30.7\%)      \\
TACO-Full                             & 1.2M (6.3\%)          & 1.2M (5.0\%)      & 4.2M (2.5\%)          & 3.5M (1.9\%)       \\ \hline
\end{tabular}
\caption{\small Graph sizes after \sys compression (lower is better)}
\label{tbl:total_graph_size}
\vspace{-3mm}
\end{table}

\begin{table}
\fontsize{6.5}{8.5}\selectfont
\centering
\begin{tabular}{llllll}
\toprule
                     &        &       Max & 75th per. & Median &   Mean \\
\midrule
  \multirow{2}{*}{Enron}    & TACO-InRow &   142,396 &  18,196 &  12,489 &  18,876 \\
                            &  TACO-Full &   700,155 &  37,286 &  18,380 &  37,963 \\
  \multirow{2}{*}{Github}   & TACO-InRow & 1,693,698 &  42,728 &  19,704 & 45,303 \\
                            &  TACO-Full & 3,139,011 &  75,553 &  31,608 & 78,633 \\
\bottomrule
\end{tabular}
\caption{\small The num. of edges reduced by \sys (higher is better)}
\label{tbl:total_edge_reduction}
\vspace{-7mm}
\end{table}

Our experiments address the following research questions:
\begin{tightitemize}
    \item How much do \sys's predefined patterns reduce formula graph sizes for real-world spreadsheets?
    (\pSection~\ref{sec:exp_graph_size})
    
    \item How much time does \sys take to build, query, and maintain a formula graph compared to \nocomp, an approach specialized for formula graph compression\rone{, and a baseline that implements the formula graph in a graph database}? (\pSection~\ref{sec:exp_nocomp} and \pSection~\ref{sec:exp_antifreeze})
    
    \item How much faster does \sys query a formula graph compared to a commercial spreadsheet system \rone{and a baseline from an open-source spreadsheet system}?
    (\pSection~\ref{sec:exp_commercial})
    
\end{tightitemize}

\subsection{Prototype, Benchmark, and Configurations}
\label{sec:benchmark_config}

\stitle{Prototype}
\sys is implemented as a Java library. 
It takes an xls or xlsx file as input, 
leverages the POI library~\cite{poi} to parse it, 
and builds a compressed formula graph 
for the parsed dependencies.\ptr{The compressed formula graph is 
implemented using an adjacency list. 
We build an R-Tree~\cite{r-tree} 
on the vertices of the formula graph  
to quickly find vertices that overlap with a given range.}
\sys provides interfaces of finding \sdeps or \sprecs of a range, 
and adding or deleting a dependency. 
\sys is integrated into 
\ds~\cite{demo1, datamodels, dataspread-demo2}, 
an open-source spreadsheet system. 
\ds returns control to users after it has identified all of the 
\sdeps of an update and hides them; 
so finding \sdeps of an update is 
the bottleneck for returning control to users. 
In \ds, a formula graph is used to find the \sdeps of an update 
and \sys acts as a drop-in replacement for this formula graph. 
\rthree{\sys can also be integrated into other spreadsheet systems, 
such as LibreCalc or MS Excel, 
to accelerate updating spreadsheets  
since these spreadsheets system similarly adopt 
formula graphs to track formula dependencies~\cite{calc-dep-graph, excel-dep-graph}. 
In addition, \sys can be used by third-party tools to analyze 
and trace formula dependencies. 
We have additionally implemented a plug-in using \sys to 
help users efficiently trace 
formula dependencies in Excel~\cite{taco-lens}.}


\stitle{Benchmark}
Our tests are based on two real-world spreadsheet datasets. 
The first one is the Enron dataset~\cite{Enron} 
with 17K xls files. 
We focus on the large spreadsheets 
(i.e., with no less than 10K dependencies) 
that do not cause exceptions (e.g., those requiring passwords), 
and are left with 593 xls files. 
Since the Enron dataset includes only xls files, 
we further crawl 7.8K xlsx files from Github
that are larger than 10 KB\footnote{Xlsx files, unlike xls files, 
support larger spreadsheets (e.g., the row limits for xlsx 
and xls files are 1M and 66K, respectively.)}. 
We focus on large spreadsheets and skip the erroneous ones, 
and get 2,238 xlsx files. 
In total, we test 2,831 files.


\begin{table}[!t]
\fontsize{6.5}{8.5}\selectfont
\centering
\begin{tabular}{llllll}
\toprule
                     &        &     Min & 25th per. & Median &   Mean \\
\midrule
  \multirow{2}{*}{Enron}    & TACO-InRow & 0.0042\% &    6.32\% & 39.81\% & 42.27\% \\
                            &  TACO-Full & 0.0042\% &     0.47\% & 1.93\% & 7.37\% \\
  \multirow{2}{*}{Github}   & TACO-InRow & 0.0005\% &     0.10\% & 17.45\% & 36.48\% \\
                            &  TACO-Full & 0.0005\% &     0.03\% &  0.19\% & 3.44\% \\
\bottomrule
\end{tabular}
\caption{\small Remaining edges after compression (lower is better)}
\label{tbl:total_edge_fraction}
\vspace{-3mm}
\end{table}

\begin{table}[!t]
\fontsize{6.5}{8.5}\selectfont
\centering
\begin{tabular}{lllll}
\toprule
 Pattern & Enron Total & Enron Max & Github Total & Github Max \\
\midrule
      RR &  17,412,246 &   525,026 &  141,876,182 &  2,094,936 \\
      RF &       1,880 &     1,413 &       13,361 &      9,999 \\
      FR &     150,845 &    13,815 &      178,609 &     39,008 \\
      FF &   3,844,351 &   174,948 &   24,784,621 &  1,043,702 \\
RR-Chain &     566,348 &    24,596 &    5,867,728 &    399,996 \\
\bottomrule
\end{tabular}
\caption{\small Num. of edges reduced by each pattern (higher is better)}
\label{tbl:pattern_edge_reduction}
\vspace{-7mm}
\end{table}

\stitle{Configurations} 
Unless otherwise specified, 
the experiments are run on a t2.2xlarge	instance 
from AWS EC2, which has 32 GB memory and 8 vCPUs, 
and uses Ubuntu 22.04 as the OS.
We use a single thread, and run each test three times 
and report the average number. 
We configure the POI library to load spreadsheets by columns. 

\begin{figure*}[t]
  \begin{tabular}{cccc}
      \begin{minipage}[b]{0.24\linewidth}
        \centering
        \includegraphics[width=\linewidth]{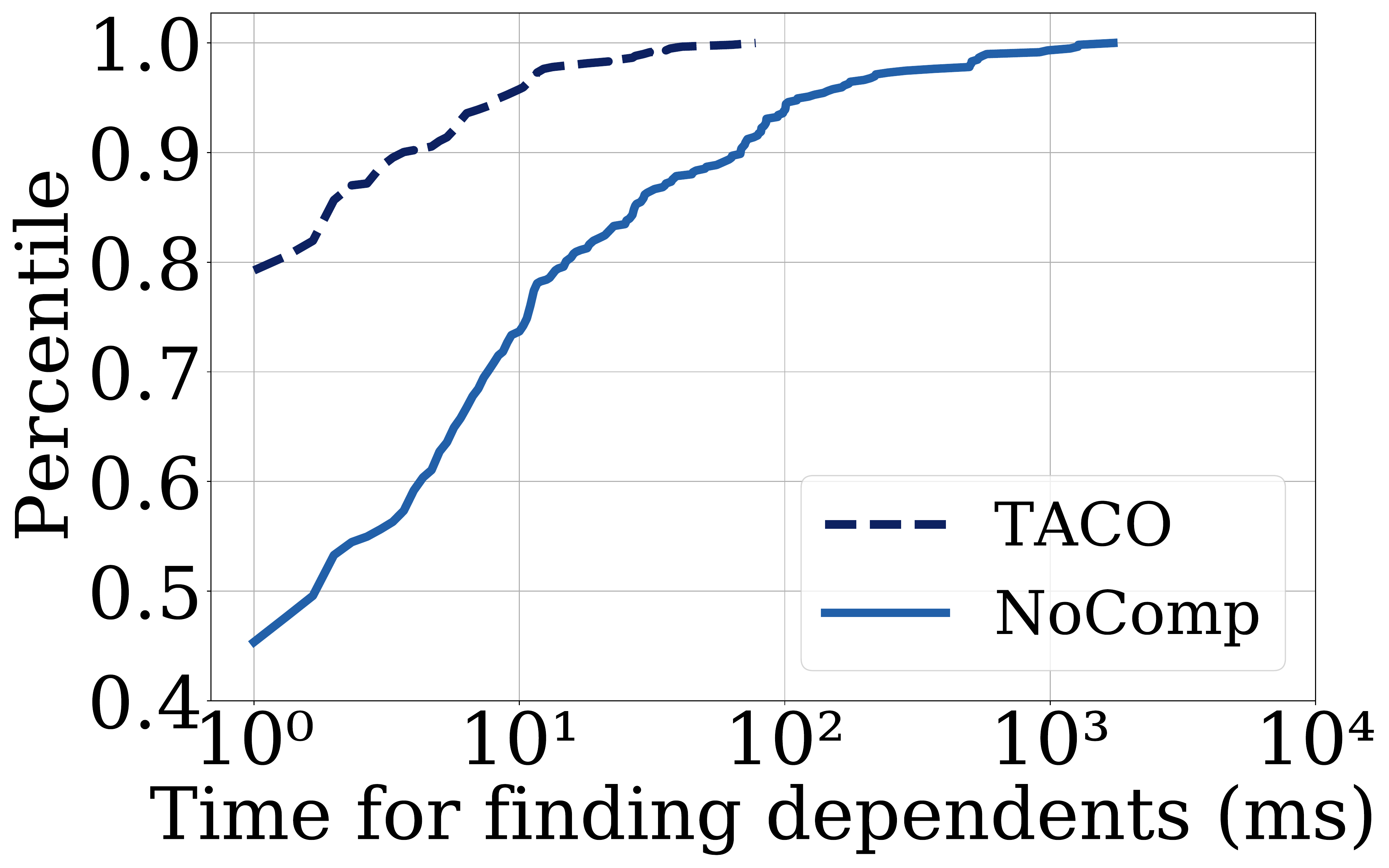}
        \subcaption{\footnotesize \maxDep (Enron)}
        \label{fig:exp_lookup_max_enron}
      \end{minipage} 

      \begin{minipage}[b]{0.24\linewidth}
        \centering
        \includegraphics[width=\linewidth]{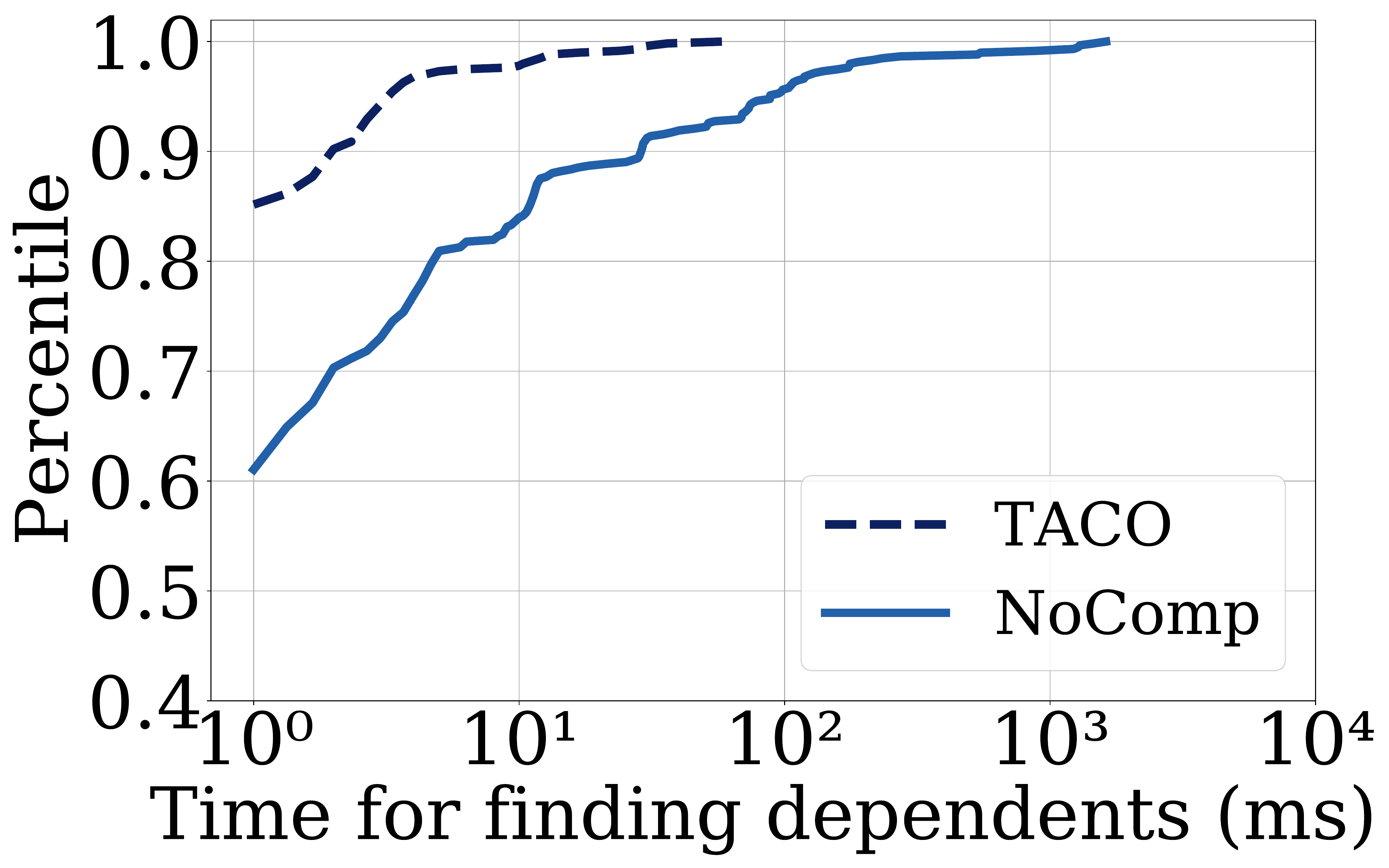}
        \subcaption{\footnotesize \longPath (Enron)}
        \label{fig:exp_lookup_long_enron}
      \end{minipage}
      
      \begin{minipage}[b]{0.24\linewidth}
        \centering
        \includegraphics[width=\linewidth]{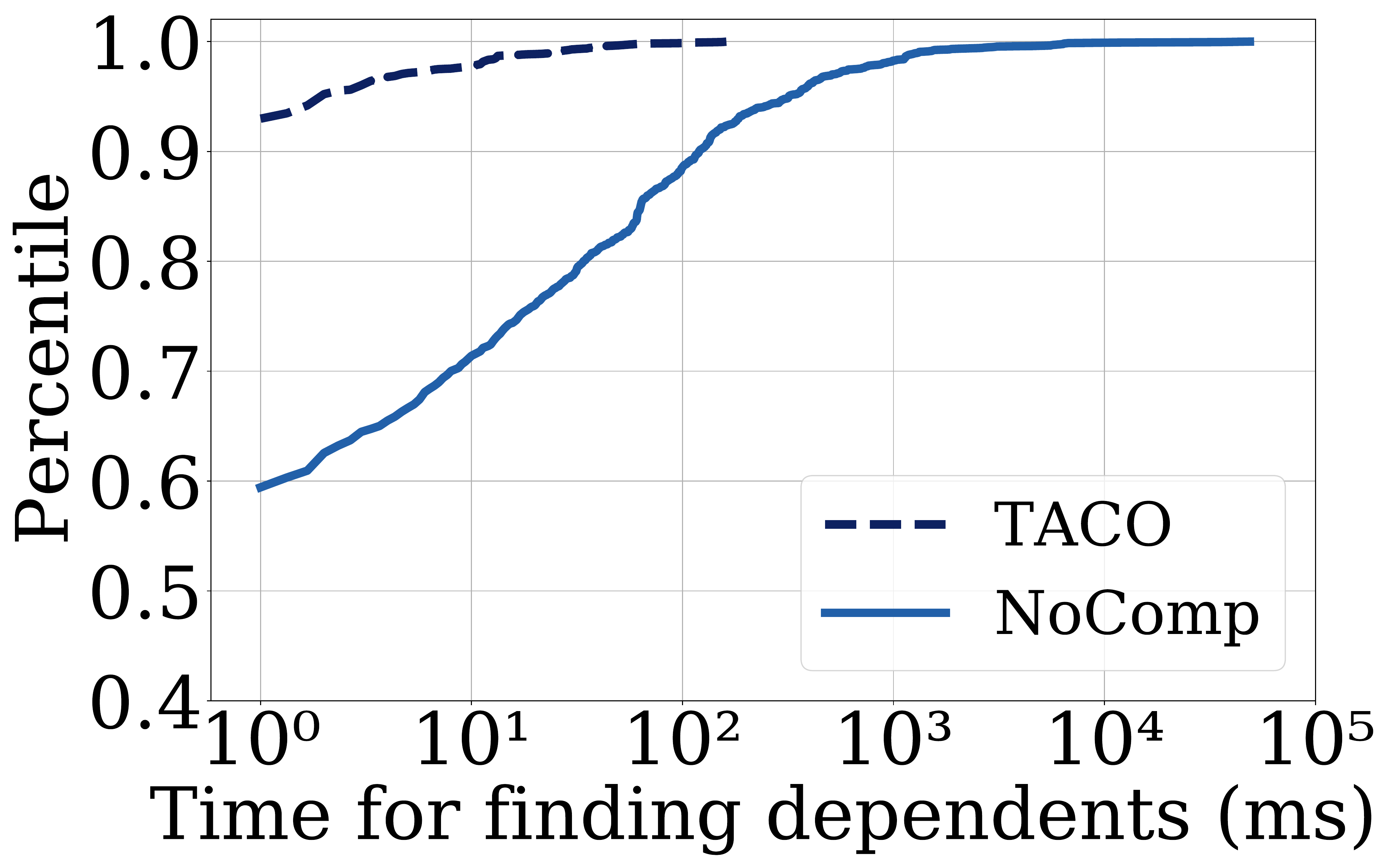}
        \subcaption{\footnotesize \maxDep (Github)}
        \label{fig:exp_lookup_max_github}
      \end{minipage}
      
      \begin{minipage}[b]{0.24\linewidth}
        \centering
        \includegraphics[width=\linewidth]{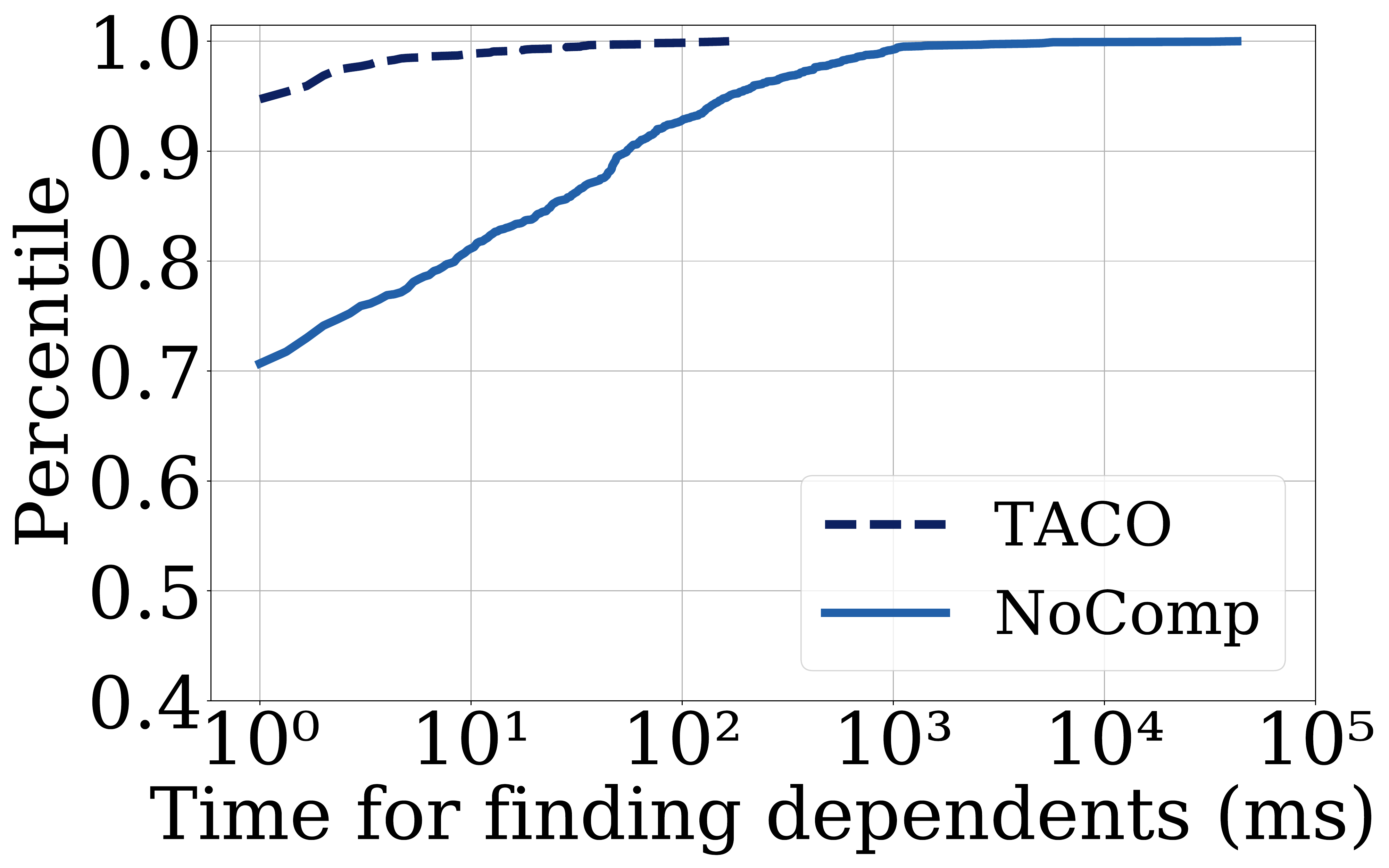}
        \subcaption{\footnotesize \longPath (Github)}
        \label{fig:exp_lookup_long_github}
      \end{minipage}
  \end{tabular}
  \vspace{-3mm}
  \caption{\small CDFs for the time for finding \sdeps}
  \label{fig:exp_lookup}
  \vspace{-4mm}
\end{figure*}

\begin{figure*}[!t]
\centering
\begin{tabular}{cc}
   \begin{minipage}{0.48\textwidth}
    \begin{tabular}{cc}
      \begin{minipage}[b]{0.5\linewidth}
        \centering
        \includegraphics[width=\linewidth]{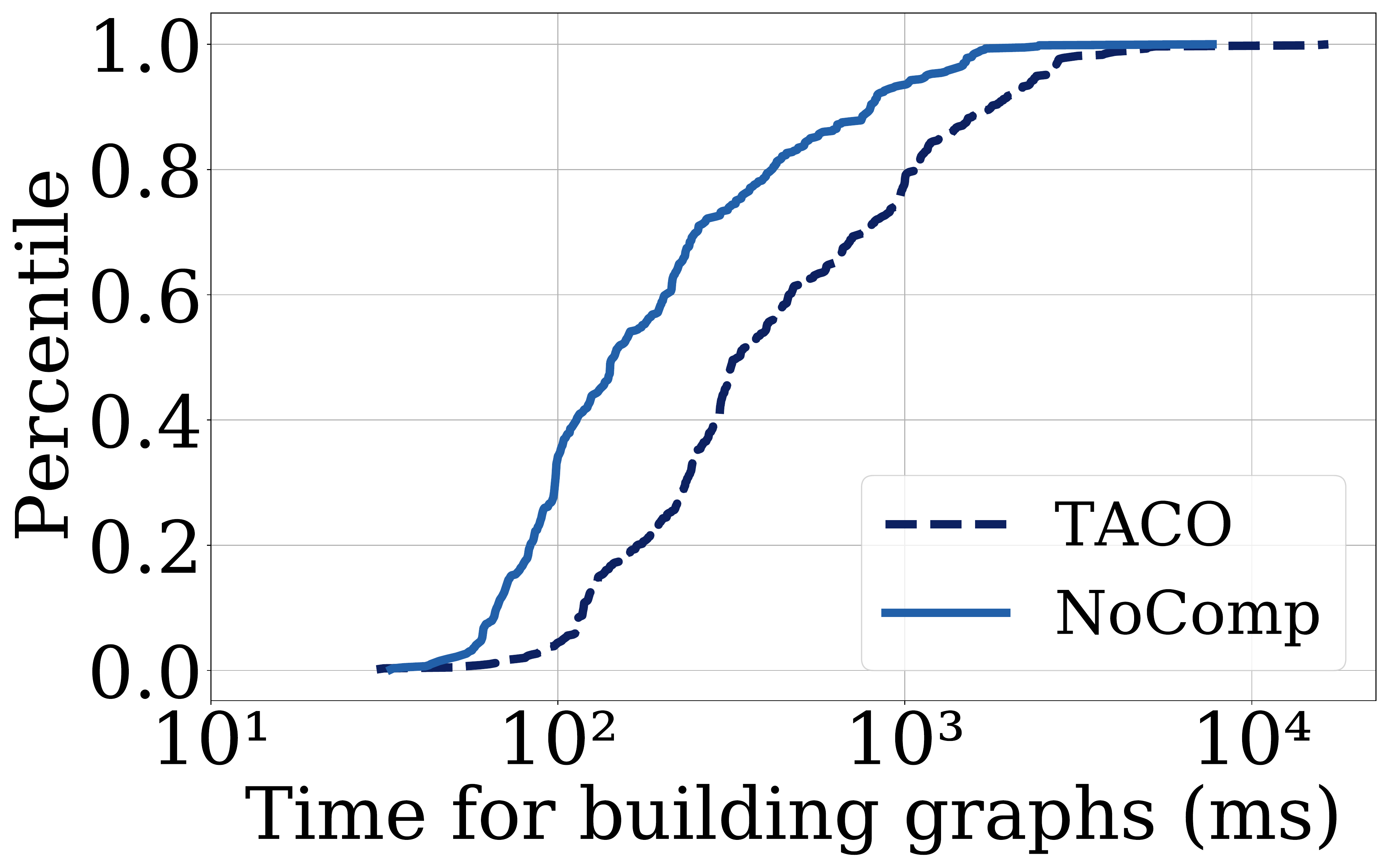}
        \subcaption{\footnotesize Enron}
        \label{fig:exp_build_enron}
      \end{minipage}
      \begin{minipage}[b]{0.5\linewidth}
        \centering
        \includegraphics[width=\linewidth]{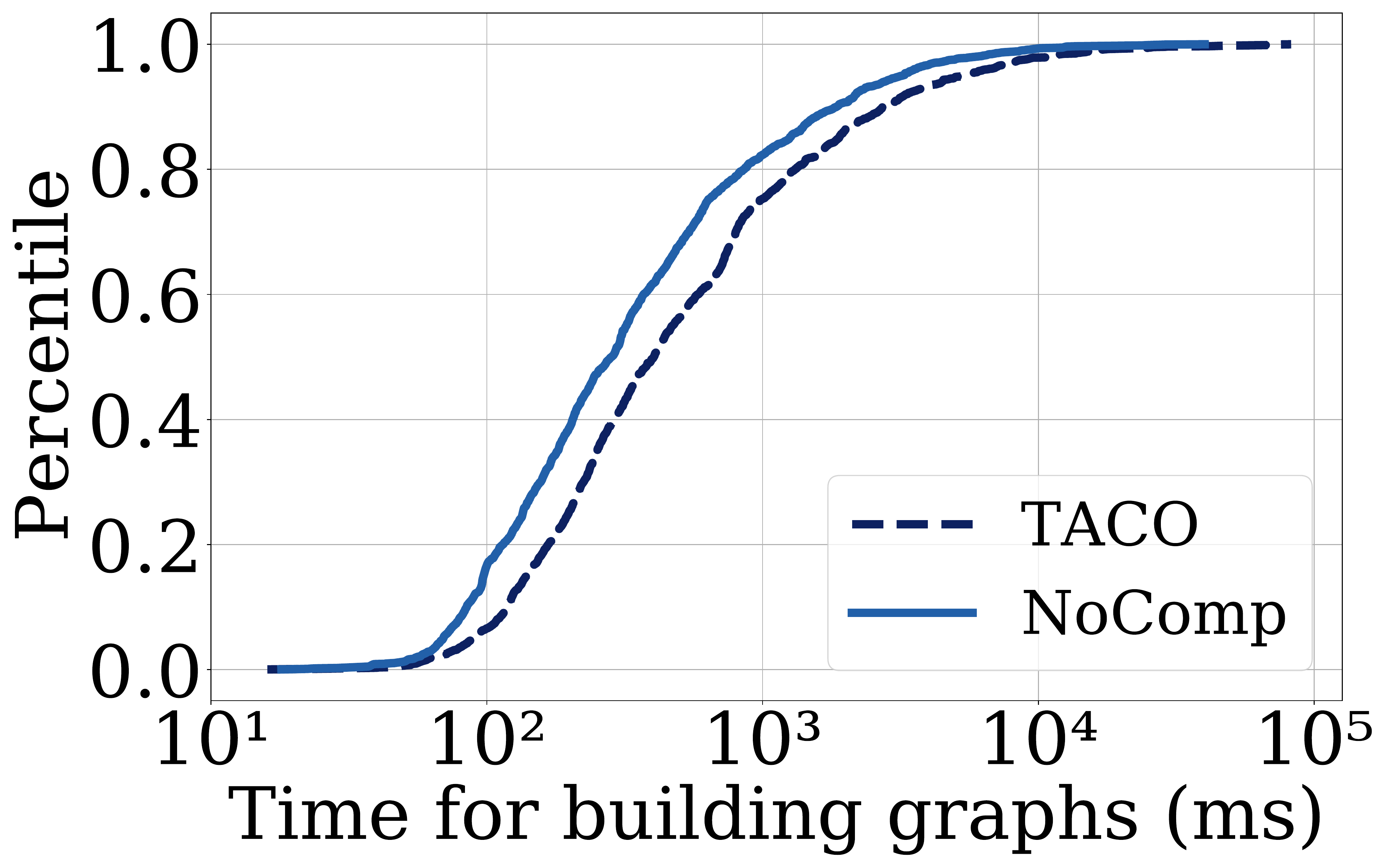}
        \subcaption{\footnotesize Github}
        \label{fig:exp_build_github}
      \end{minipage} 
    \end{tabular}
    \vspace{-3mm}
    \caption{\small CDFs for the time for building formula graphs}
  \label{fig:exp_build}
  \vspace{-7mm}
  \end{minipage}

   \begin{minipage}{0.48\textwidth}
    \begin{tabular}{cc}
      \begin{minipage}[b]{0.5\linewidth}
        \centering
        \includegraphics[width=\linewidth]{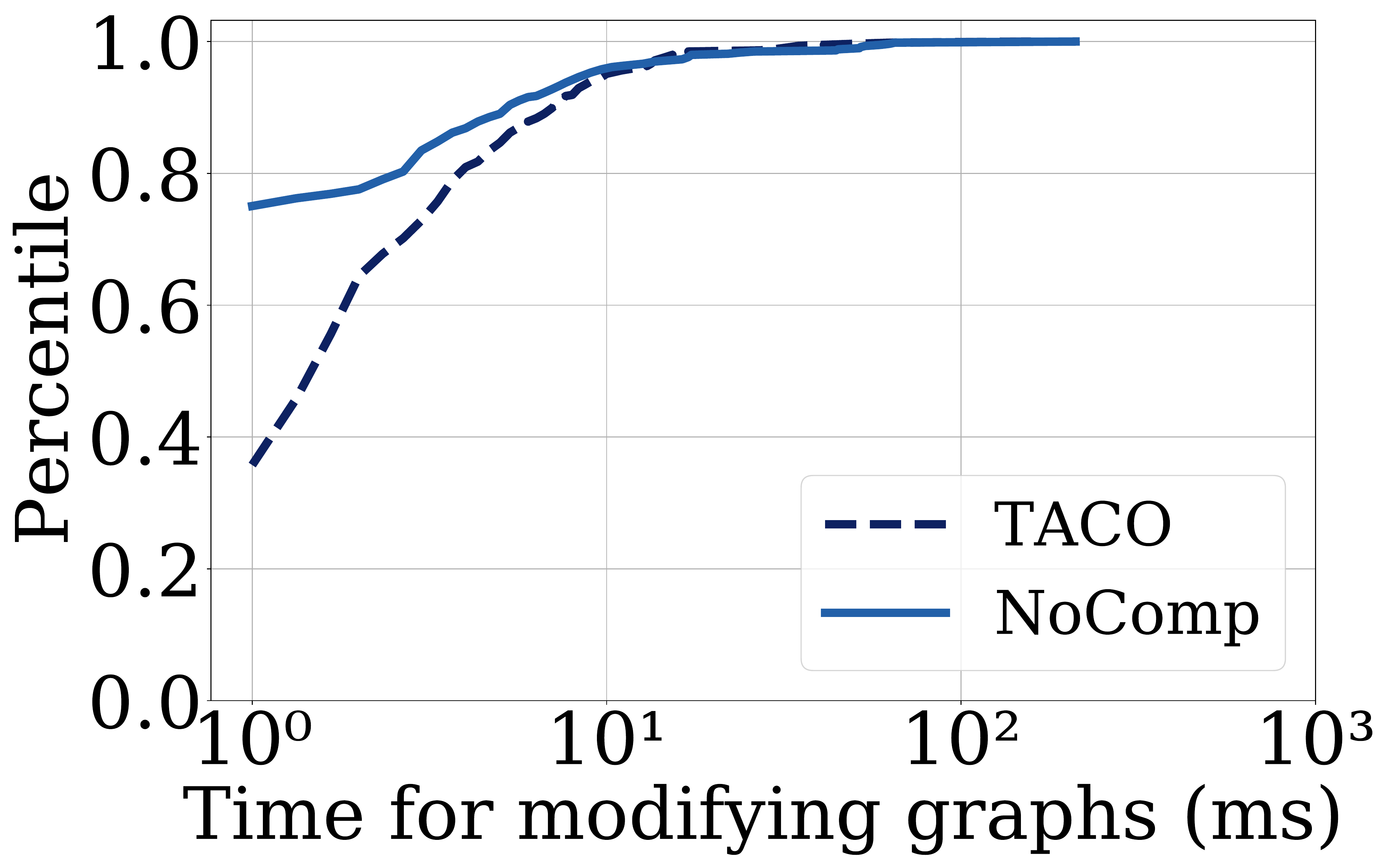}
        \subcaption{\footnotesize Enron}
        \label{fig:exp_modify_enron}
      \end{minipage}
      \begin{minipage}[b]{0.5\linewidth}
        \centering
        \includegraphics[width=\linewidth]{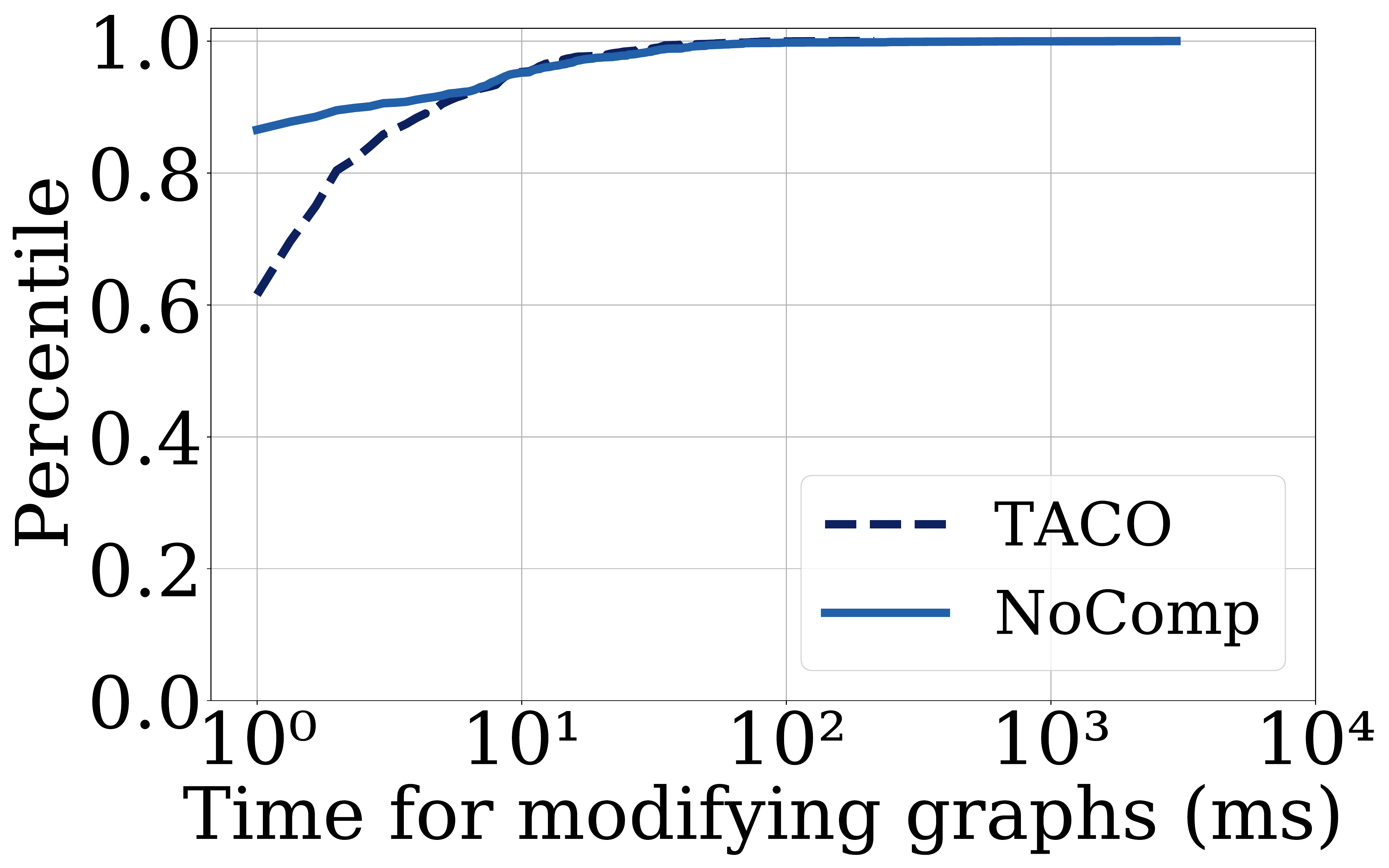}
        \subcaption{\footnotesize Github}
        \label{fig:exp_modify_github}
      \end{minipage} 
    \end{tabular}
    \vspace{-3mm}
    \caption{\small CDFs for the time for modifying formula graphs}
  \label{fig:exp_modify}
  \vspace{-4mm}
  \end{minipage}

\end{tabular}
\end{figure*} 

\subsection{Compressed Formula Graph Sizes}
\label{sec:exp_graph_size}

We first test the effectiveness of \sys in reducing the graph sizes. 
We test two variants: \inrow and \full. 
\inrow only compresses 
adjacent column formulae that reference ranges in the same row, 
using \rr to perform the compression. 
This approach 
captures the pattern of derived columns, 
where a subset of columns are computed using the remaining, 
which is common in data science and feature engineering (e.g., storing normalized versions of values in a column as a new column, or extracting substrings of an existing column and storing them as a new column).
\full considers any formulae and adopts all predefined patterns.

\stitle{Overall effectiveness of reducing graph sizes}
We first report the total number of vertices and edges 
of the compressed and uncompressed formula graphs across all the files 
in Enron and Github, respectively, in Table~\ref{tbl:total_graph_size}. 
The uncompressed graphs are built using \nocomp, as discussed in \pSection~\ref{sec:complexity}. 
Both \inrow and \full significantly reduce 
the total number of vertices and edges compared to \nocomp. 
For example, \full reduces the number of edges 
in Github from 179.8M to 3.5M. 
In addition, \full has much smaller graph sizes 
than \inrow (e.g., 3.5M vs. 55.2M edges for Github), 
which shows that many formulae reference 
different rows and \full can efficiently 
compress these complex cases that \inrow does not consider. 

To further understand the effectiveness of \sys, 
we compute two additional metrics for each spreadsheet's uncompressed formula graph 
$G^\prime(E^\prime, V^\prime)$ and compressed formula graph $G(E, V)$: 
the number of edges reduced by the \sys (i.e., $|E^\prime| - |E|$) 
and the fraction of the number of compressed edges 
compared to the uncompressed ones 
(i.e., $\frac{|E|}{|E^\prime|}$). 
Table~\ref{tbl:total_edge_reduction} reports the max, 75th percentile, 
median, and mean value of the number of reduced edges 
across all files for both datasets. 
Table~\ref{tbl:total_edge_fraction} reports the min, 25th percentile, 
median, and mean value of the edge fraction after compression.

Table~\ref{tbl:total_edge_reduction} shows that 
\full can reduce the number of edges by {\bf up to 700K and 3.1M 
in a single spreadsheet} for Enron and Github, respectively. 
The average edge reduction by \full is 
38K and 79K for the two datasets. 
Table~\ref{tbl:total_edge_fraction} shows that 
the average edge fractions after compression by \full 
are \textbf{as low as 7.4\% and 3.4\%} for Enron and Github, respectively. 
These results show that \sys can effectively reduce formula graph sizes 
of real spreadsheets. 

\stitle{Effectiveness of \sys patterns}
Next, we evaluate the effectiveness of each \sys pattern 
in reducing the number of edges. 
Recall that a partition of edges $E^{\prime}_i$ 
in the original uncompressed graph $G^\prime(E^\prime, V^\prime)$ 
corresponds to one compressed edge $e_i$ in $G(E, V)$ in \sys. 
So the number of reduced edges of a pattern $p$ 
in $G$ is computed as: $\sum_{e_i \in E}[e_i.meta.p=p](|E^{\prime}_i| - 1)$, 
where $[e_i.meta.p=p]$ considers the compressed edges for the pattern $p$, 
$E^{\prime}_i$ is the set of edges that are compressed into $e_i$, 
and $|E^{\prime}_i| - 1$ is the number of reduced edges by $e_i$. 
We compute the above metric for each pattern, 
and report the total and maximum number of reduced 
edges across the tested spreadsheets. 

The results in Table~\ref{tbl:pattern_edge_reduction} 
show that \rr and \ff compress the most edges. 
The number of edges reduced by \rr 
is more than 17.4M and 141.9M for Enron and Github, respectively. 
\ff reduces around 3.8M and 24.8M edges in total 
for the two respective datasets. 
Other patterns also reduce a significant number of edges 
in some spreadsheets. 
For example, in the Github dataset 
\rrChain reduces the number of edges up to 
around 400K for a single spreadsheet. 
\fr and \rf, while not as common, 
can reduce up to around 39K and 
10K edges for a single spreadsheet, respectively. 
These results show that \sys's patterns are prevalent 
in real spreadsheets and can significantly reduce graph sizes.

\begin{figure*}[t]
  \begin{tabular}{cc}
      \begin{minipage}[b]{0.48\linewidth}
              \centering
        \includegraphics[width=\linewidth]{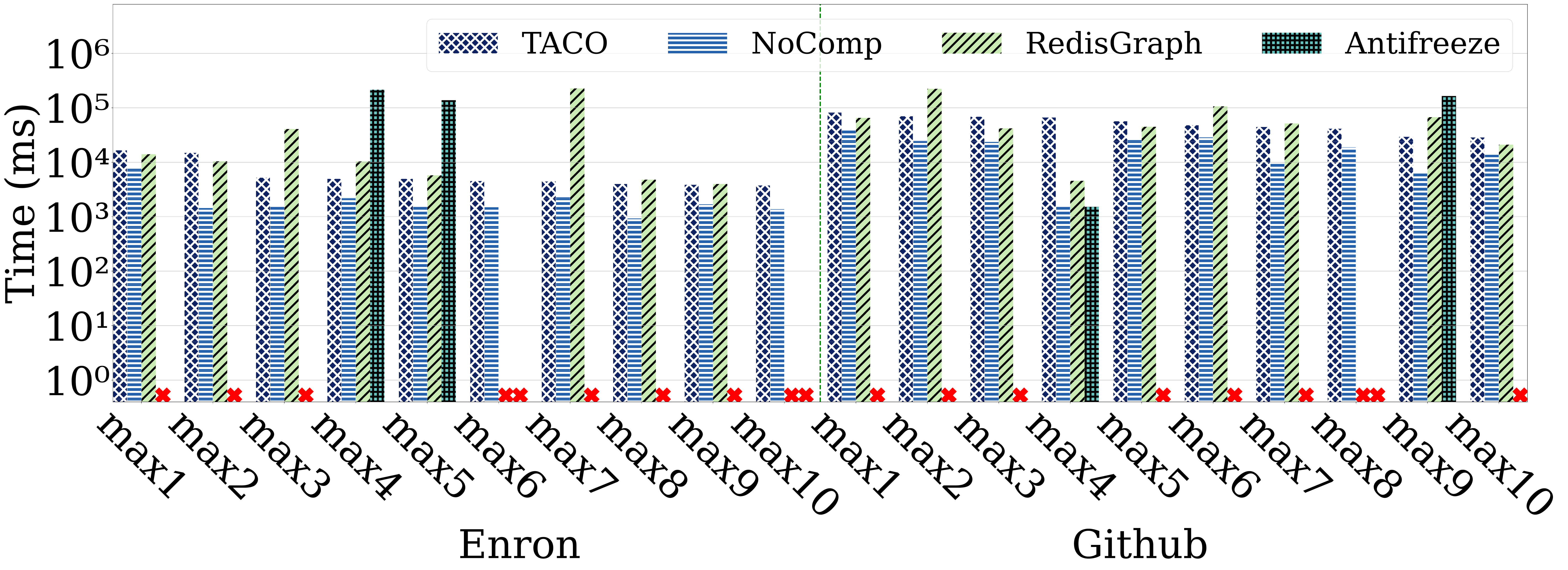}
        \vspace{-7mm}
        \caption{\footnotesize \rone{Latency on building graphs (\anti and \redis)}}
        \label{fig:exp_antifreeze_build}
      \end{minipage} 
      \begin{minipage}[b]{0.48\linewidth}
        \centering
        \includegraphics[width=\linewidth]{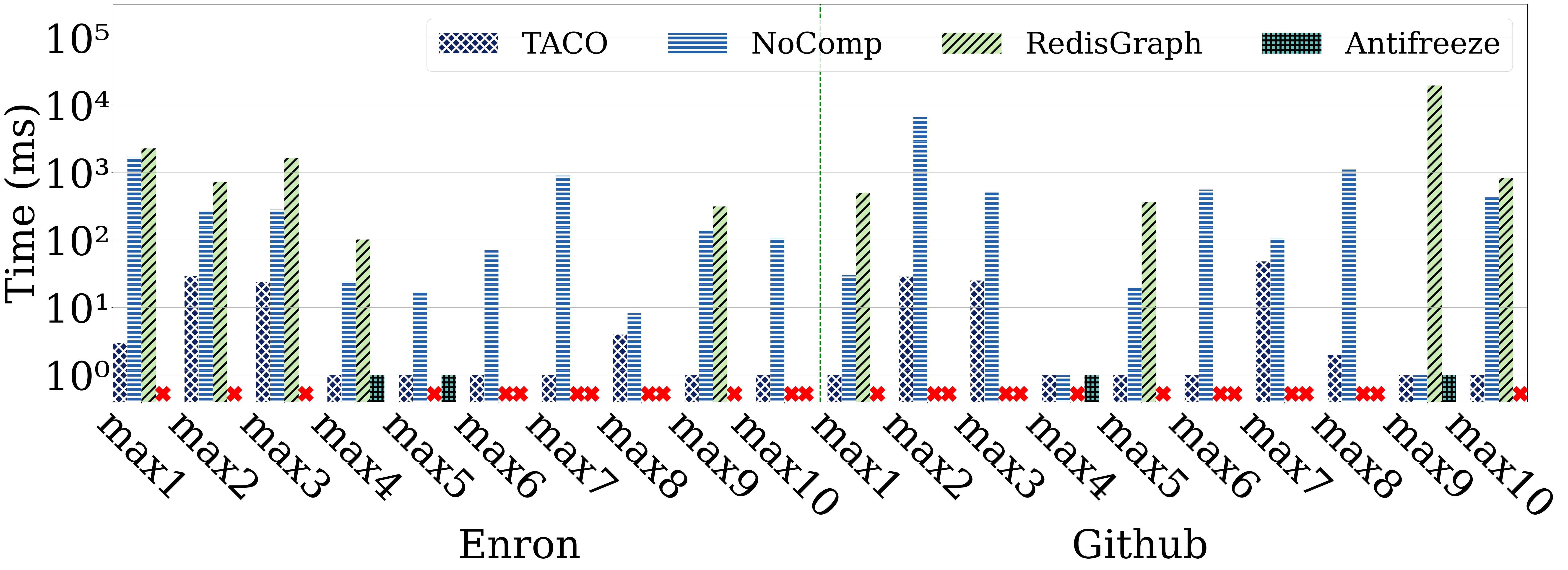}
        \vspace{-7mm}
        \caption{\footnotesize \rone{Latency on finding \sdeps (\anti and \redis)}}
        \label{fig:exp_antifreeze_lookup}
      \end{minipage} \\
      
      \begin{minipage}[b]{0.48\linewidth}
        \centering
        \includegraphics[width=\linewidth]{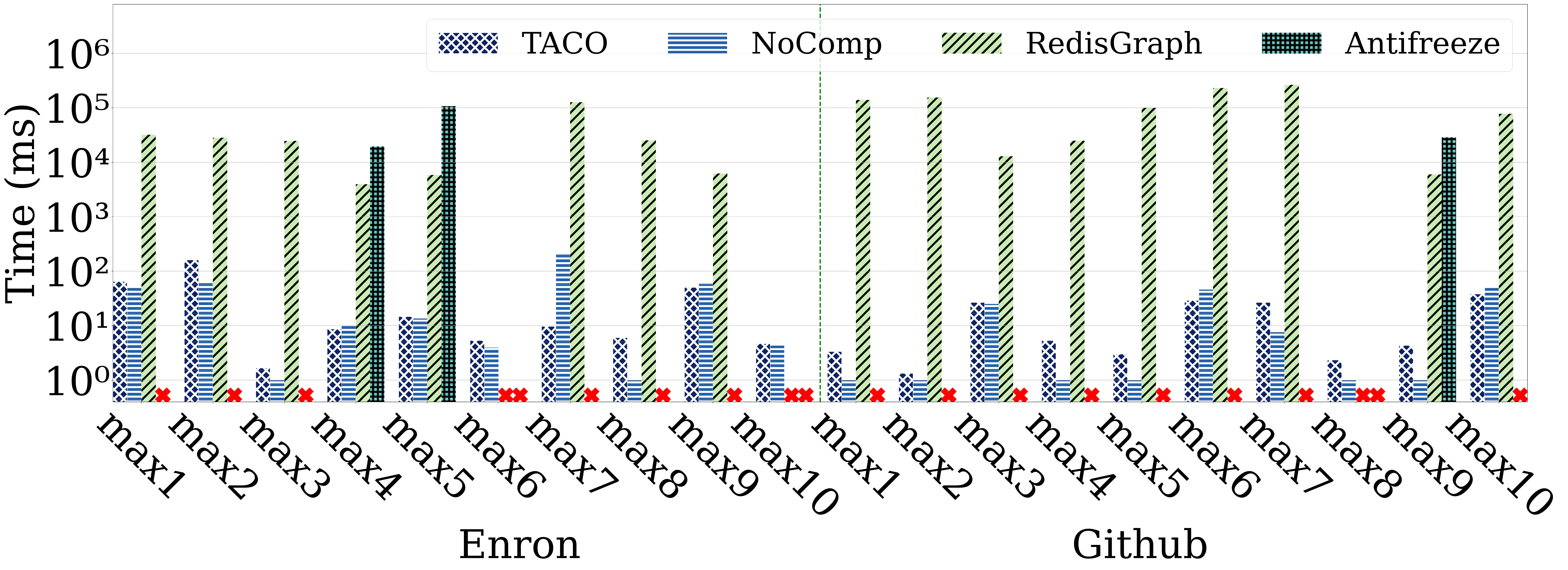}
        \vspace{-7mm}
        \caption{\footnotesize \rone{Latency on modifying graphs (\anti and \redis)}}
        \label{fig:exp_antifreeze_modify}
      \end{minipage}
      
      \begin{minipage}[b]{0.48\linewidth}
        \centering
        \includegraphics[width=\linewidth]{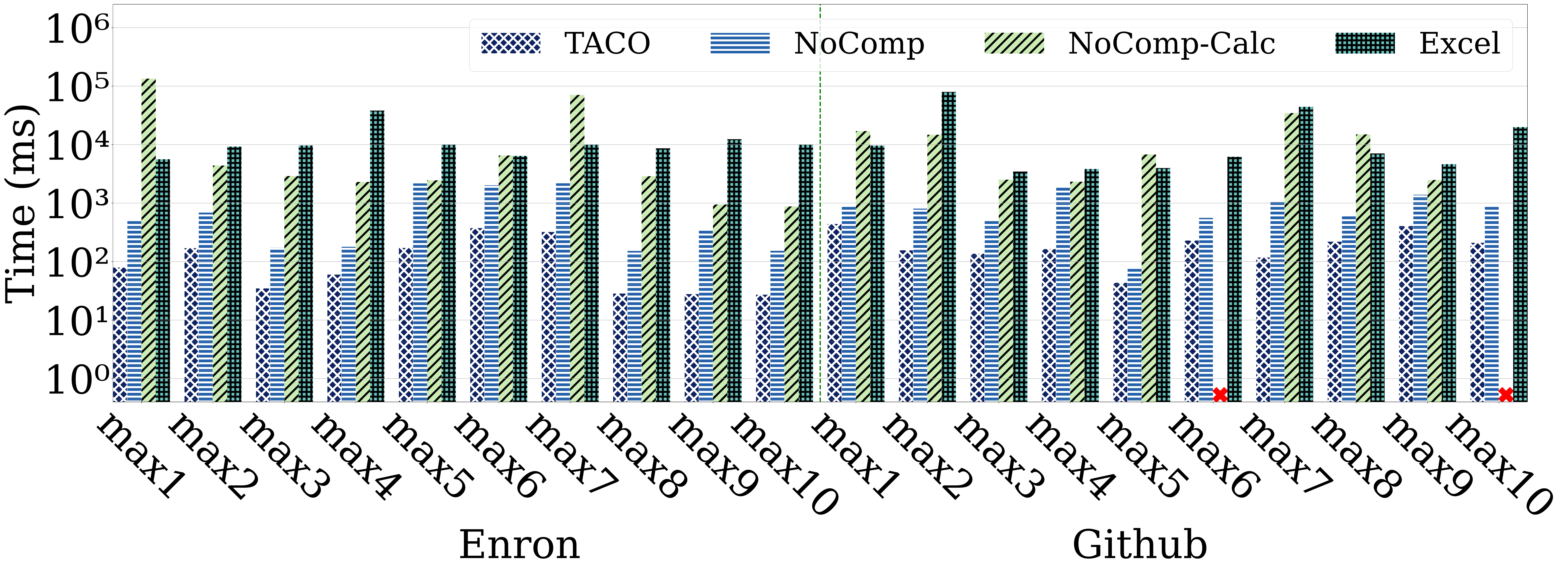}
        \vspace{-7mm}
        \caption{\footnotesize \rone{Latency on finding \sdeps (Excel and \calc)}}
        \label{fig:exp_excel}
      \end{minipage}

  \end{tabular}
  \vspace{-7mm}
\end{figure*}

\subsection{Performance Comparison with \nocomp}
\label{sec:exp_nocomp} 
We now compare the performance of \sys and \nocomp, 
including the time for finding \sdeps, building formula graphs, 
and modifying formula graphs. 
We focus on finding \sdeps because 
finding \sprecs is the dual problem. 

\stitle{Finding \sdeps}
For each spreadsheet we test two cases: 
finding \sdeps for the cell that has the maximum number of 
\sdeps (denoted as the \maxDep case) and the cell that has 
the longest path (denoted as the \longPath case) 
in the uncompressed graph. 
\rone{Recall that we consider both direct and 
indirect \sdeps of a cell for finding \sdeps.}
\pFigure~\ref{fig:exp_lookup} reports the CDFs for the time of 
finding \sdeps for the two cases in the two datasets. 
We see that \sys has much smaller execution time 
for finding \sdeps than \nocomp. 
For Enron and Github datasets, 
\sys's maximum execution time for finding \sdeps 
is 78~ms and 167~ms, respectively, 
while \nocomp's maximum execution time is 
1,730~ms and 48,889~ms, respectively. 
Across all of the tested spreadsheets, 
\textbf{the speedup of \sys over \nocomp is up to 34,972$\times$}. 

\stitle{Building and modifying formula graphs}
We also test the time for building and 
modifying formula graphs for the two datasets. 
To modify a formula graph of a spreadsheet, we remove the content 
of a column of 1K cells starting from the cell that has the 
most \sdeps. 

\pFigure~\ref{fig:exp_build} reports the CDFs for 
the time of building formula graphs in two datasets. 
We see \sys takes more time to build the formula graphs 
compared to \nocomp due to the compression overhead. 
For Enron, the longest time for building a formula graph 
for \sys and \nocomp is 16,626~ms and 7,704~ms, respectively. 
For Github, this number for \sys and \nocomp 
is 82,567~ms and 40,103~ms, respectively. 
We believe this overhead is acceptable because  
building formula graphs only happens once 
when we load the spreadsheet, 
and it can be executed in the background asynchronously 
and will not be on the critical path of users interacting with the system. 
\pFigure~\ref{fig:exp_modify} reports the CDFs 
for the time taken to modify formula graphs. 
We see that for the easy cases composing the first 90\% 
with less than 10~ms, 
\sys takes more time to modify formula graphs than \nocomp. 
For the harder cases, \sys takes less time than \nocomp, 
which is consistent with our complexity analysis. 
For example, for Github, the 99th percentiles for \sys and  
\nocomp are 33~ms and 41~ms, respectively.

\subsection{Performance Comparison with \anti and \redis}
\label{sec:exp_antifreeze}
We now compare \sys with an approach that 
involves compressing formula graphs, \anti~\cite{Antifreeze},\footnote{Note that graph compression is one of the Antifreeze paper’s contributions; its main focus is on the asynchronous execution model, metric, and interface.}
\rone{and another approach that instantiates formula graphs 
in a graph database, \redis~\cite{redis}, without compression.}
\anti builds an uncompressed formula graph 
for the input dependencies, pre-computes the \sdeps for each cell, 
compresses the \sdeps for each cell via bounding ranges, 
and stores each cell along with the compressed \ldeps 
in a look-up table. 
If formula cells are changed, it modifies the uncompressed graph 
and builds the look-up table from scratch. 
The number of bounding ranges is set to 20,  
as in the original paper~\cite{Antifreeze}. 
\rone{To store formula graphs in \redis, 
we decompose each edge that involves a range 
in the original formula graph 
into multiple edges that only involve cells, 
since \redis and other graph databases do not 
support finding overlapping vertices for an input range. 
For example, an edge $\code{A1:A2} \rightarrow \code{B1}$ 
is decomposed into two edges $\code{A1} \rightarrow \code{B1}$ 
and $\code{A2} \rightarrow \code{B1}$. 
To efficiently load a formula graph in \redis, we 
write these decomposed edges into a CSV file and adopt 
\redis's bulk-load tool~\cite{redis-bulk-load} 
to load this file rather than inserting each edge one by one. 
We use Cypher~\cite{cypher}, a declarative graph query language, 
to query and maintain the graph.}

For this test, we choose the top 10 spreadsheets for which \sys has 
the longest time for building formula graphs from each dataset
\footnote{We have tried the spreadsheets 
where \sys has the longest time 
for finding \sdeps, but \anti cannot finish for any of them.}.
We rename top 10 spreadsheets to $max_i$, 
where $i$ represents the order. 
\rone{If the time for building a formula graph 
is greater than 300 seconds, 
this test is regarded as \textit{did not finish} (DNF). 
For \redis, we mark the test for 
finding dependents as DNF if it cannot finish within 60 seconds 
since we observed that the memory consumption grows quickly 
when finding dependents in \redis. 
This is because \redis does not efficiently 
optimize the declarative Cypher query 
and needs to search one edge multiple times. 
In the experiment figures, we use a red \textbf{X} to 
represent DNF.}

\pFigure~\ref{fig:exp_antifreeze_build}-\ref{fig:exp_antifreeze_modify} 
report the time for finding \sdeps of the cell that has 
the maximum number of \sdeps, 
and the time for building and modifying a formula graph. 
We see that \anti only finishes building the 
compressed formula graph for 4 out of the 20 spreadsheets 
and so its other numbers are not reported (marked 
with \textbf{X}).
For the spreadsheets where \anti can finish the tests, 
\sys has the same execution time for finding \sdeps 
as \anti, and has much smaller execution time for 
building and modifying a formula graph than \anti. 
\rone{\redis cannot finish in many tests, either, 
mainly due to the large graphs that only include 
cell-to-cell edges. 
Among the tests where \redis finishes, 
\sys has much smaller execution time than \redis 
in most of these tests. Specifically, the speedup of \sys over \redis 
on finding dependents is up to 19,555$\times$.}

\subsection{Performance Comparison with Excel and \calc}
\label{sec:exp_commercial}

We now compare \sys's performance of 
finding \sdeps with Excel \rone{and a baseline derived
from OpenOffice Calc~\cite{calc-dep-graph}, denoted \calc}. 
For Excel, we test the VBA API for finding 
the \sdeps of a cell~\cite{excel-vba-dependents}.
\rone{For \calc, we implemented it 
based on a document that describes the design of formula 
graphs in OpenOffice Calc~\cite{calc-dep-graph}. 
Similar to \nocomp, this baseline does not compress dependencies. 
The difference from \nocomp is that \calc does not 
use an R-Tree to find the vertices in formula graphs 
that overlap with an input range. 
Instead, it pre-partitions 
the spreadsheet space into containers, 
stores overlapping ranges in each container, 
and uses containers to find the overlapping vertices.}
We use top 10 spreadsheet files 
for which \sys spends the most time for finding \sdeps 
in each dataset from \pSection~\ref{sec:exp_nocomp}. 
We rename the 10 spreadsheets to $max_i$, where $i$ 
represents the order. 
\rone{In each spreadsheet, we test the time for finding \sdeps 
of the cell that has the maximum number of dependents. 
If a test cannot finish within 300 seconds, it is marked as a red \textbf{X}.}
These experiments are done on a laptop 
that has one Intel Core i5 CPU with 4 physical cores 
and 8 GB of memory, and uses Windows 10 as the OS. 

The results in \pFigure~\ref{fig:exp_excel}  
show that \sys is much faster than Excel in all cases. 
The longest time for finding dependents for \sys 
and Excel is 442~ms and 79,761~ms, respectively. 
The speedup of \sys over Excel is up to 632$\times$ 
(i.e., max4 from Enron). 
It is surprising that Excel takes longer time for finding \sdeps
than \nocomp in all cases. One possible reason is that 
Excel compresses formula graphs to reduce memory consumption, 
which introduces the overhead of decompression 
when the formula graphs are used for finding \sdeps. 
\rthree{We note that since Excel is a complex system, 
it may have the overhead that \sys does not. 
It is also possible that Excel is optimized for 
other scenarios by sacrificing the performance 
of finding dependents.} 
\rone{For \calc, it cannot finish in two cases. 
For the other cases, \sys is much faster than \calc 
and the speedup of \sys over \calc is up to 1,682$\times$.}
\section{Related Work}
\label{sec:related}

\sys is related to formula computation, 
graph compression, column-oriented databases, and scalable spreadsheets. 

\stitle{Formula computation in spreadsheets}
There has been some work on improving the interactivity 
of spreadsheets during updates. 
Excel~\cite{excel-dep-graph} and 
other spreadsheet systems~\cite{sestoft2014spreadsheet, zk_spreadsheet, calc-dep-graph} 
track \sdeps of formula cells 
to quickly identify the cells impacted 
by an update and recalculate them. 
\ds approaches this problem using asynchronous execution
~\cite{marcus2011crowdsourced, parameswaran2012deco, brin1998anatomy}. 
It uses the formula graph to identify the impacted formula cells and mark them dirty, 
return control to users immediately, 
and calculate the dirty cells asynchronously~\cite{Antifreeze}. 
Unlike \sys, none of these approaches leverage tabular locality to 
compress formula graphs. 
In addition, \sys is orthogonal to the execution models 
and can be integrated into an existing spreadsheet system 
to improve interactivity. 

\stitle{Graph compression}
Graph compression has been studied in many scenarios, 
such as in the Web~\cite{webgraph} and social networks~\cite{socialGraphCompression}. 
A recent survey~\cite{liu2018graph} shows that
different compression methods are designed for different goals,  
including understanding the structure of a graph~\cite{cilibrasi2005clustering}, 
reducing graph sizes with bounded errors~\cite{navlakha2008graph}, or accelerating queries on graphs~\cite{maccioni2016scalable}. 
None of these papers leverage tabular locality 
and consider the spatial nature of formula graphs. 
In addition, most of them do not support directly 
querying the compressed graph.
Fan et al.~\cite{queryPreserving} support directly 
executing reachability and pattern matching queries 
on a compressed graph, but do not leverage tabular locality 
and support finding \sdeps/\sprecs.
While a recent paper proposes a compressed graph for spreadsheets~\cite{Antifreeze}, 
we showed that building such a compressed graph 
is time-consuming.  

\stitle{Column-oriented databases}
Column-oriented databases~\cite{CStore, MonetDB} employ lightweight compression 
for data in each column and execute queries directly on the 
compressed data without decompression~\cite{ColumnCompression}. 
\sys is instead designed to compress dependencies (i.e., edges) 
while column-oriented compression 
methods are used to compress columnar data. 
In addition, column-oriented databases do not consider 
decomposing complex patterns in a column of formulae into predefined patterns 
as \sys proposes. 

\stitle{Spreadsheets at scale}
Many prior papers focus on supporting 
large-scale data analysis on spreadsheets
~\cite{1010data,raman1999scalable, airtable, demo1, datamodels, witkowski_advanced_2005, witkowski_query_2005, Mondrian}. 
\ds~\cite{demo1, datamodels} adopts databases as 
a scalable back-end. 
ABC~\cite{raman1999scalable} provides a spreadsheet interface 
and uses approximate query processing techniques, 
such as online aggregation~\cite{raman1999online}, 
to quickly return results. 
Mondrian~\cite{Mondrian} maps 
spreadsheets to visual images to detect 
different regions in spreadsheets 
and extract layout templates, but 
does not consider formula dependency compression. 
\sys is different from these papers 
because it approaches the scalability 
problem using compression techniques. 

\vspace{-1mm}
\section{Conclusion}
We presented \sys---a framework that efficiently compresses 
formula graphs in spreadsheets to improve interactivity.
\sys exploits tabular locality, 
wherein cells close to each other have formulae with similar structures, 
and represents tabular locality via 
four basic and one extended patterns. 
As part of \sys, we introduce algorithms 
for building the compressed formula graph based on predefined patterns, 
querying this graph without decompression, 
and incremental maintenance. 
Our experiments show that \sys can quickly find the dependents of 
spreadsheet cells to significantly reduce the time of returning control to users 
while achieving fast graph maintenance at the same time. 

\begin{NoHyper}
\scriptsize
\bibliographystyle{IEEEtran}
\bibliography{references}
\end{NoHyper}

\end{document}